\documentclass[onefignum,onetabnum]{siamart190516}

\usepackage{amsmath}
\usepackage{subcaption}

\newcommand{\allT}{T}
\newcommand{\allD}{\xi}
\newcommand{\pop}{\text{pop}}

\newcommand{\tree}{\tau}

\usepackage{lipsum}
\usepackage{amsfonts}
\usepackage{graphicx}
\usepackage{epstopdf}
\usepackage{algorithmic}
\ifpdf
  \DeclareGraphicsExtensions{.eps,.pdf,.png,.jpg}
\else
  \DeclareGraphicsExtensions{.eps}
\fi

\newsiamremark{remark}{Remark}
\newsiamremark{hypothesis}{Hypothesis}
\crefname{hypothesis}{Hypothesis}{Hypotheses}
\newsiamthm{claim}{Claim}

\headers{Metropolized Forest Recombination}{E. Autry, D. Carter, G. Herschlag, Z. Hunter, J. Mattingly}

\title{Metropolized Forest Recombination for Monte Carlo Sampling of Graph Partitions.\thanks{Submitted to the editors 10/13/2020.
\funding{This work was funded by the NSF grant DMS-1613337}}}

\author{Eric Autry\thanks{Department of Mathematics, Duke University}
\and Daniel Carter\thanks{North Carolina School of Science and Mathematics, Durham NC (Currently Princeton University)}
\and Gregory Herschlag\thanks{Department of Mathematics, Duke University (\email{gjh@math.duke.edu})}
\and Zach Hunter\thanks{North Carolina School of Science and Mathematics, Durham NC (Currently Oxford University)}
\and Jonathan C. Mattingly\thanks{Department of Mathematics and Department of Statistical Science, Duke University (\email{jonm@math.duke.edu})}
}

\usepackage{amsopn}

\ifpdf
\hypersetup{
  pdftitle={Metropolized Forest Recombination},
  pdfauthor={E. Autry, D. Carter, G. Herschlag, Z. Hunter, J. Mattingly}
}
\fi

\begin{document}

\maketitle

\begin{abstract}
We develop a new Markov chain on graph partitions that makes relatively global moves yet is computationally feasible to be used as the proposal in the Metropolis-Hastings method. Our resulting algorithm can be made reversible and able to sample from a specified measure on partitions. Both of these properties are critical to some important applications and computational Bayesian statistics in general. Our proposal chain modifies the recently developed method called Recombination (ReCom), which draws spanning trees on joined partitions and then randomly cuts them to repartition. We improve the computational efficiency by augmenting the state space from partitions to spanning forests. The extra information accelerates the computation of the forward and backward proposal probabilities. We demonstrate this method by sampling redistricting plans and find promising convergence results on several key observables of interest.
\end{abstract}



\section{Introduction}
Graph partition problems have become important in a variety of fields including clustering and detection of cliques in social networks (e.g. \cite{fortunato2010community}), pathological and biological networks (e.g. \cite{junker2011analysis}), infrastructural networks such as roads and aiport controls (e.g. \cite{ich2006extremely,mohring2007partitioning}), image decomposition (e.g. \cite{camilus2012review}), and problems in redistricting (e.g. \cite{Chikina_Frieze_Pegden_2017,herschlag2020quantifying,deford2019redistricting}) (for a review see \cite{bulucc2016recent}).  Typically, graph partition problems involve optimizing an objective function, such as minimizing the sum of edge weights spanning two partitions \cite{bulucc2016recent} or minimizing communication loads across partitions \cite{hendrickson2000graph,bader2013graph}. \footnote{This article was originally released in October 2019 with the name ``A Merge-Split Proposal for Reversible Monte Carlo Markov Chain Sampling of Redistricting Plans'' to emphasize the imagery of merging and then splitting partitions and the initial motivation of generating  Redistricting Plans. We changed the name to better indicate the connection to the ReCom algorithm and the feasibility of our algorithm to mesh with the   Metropolis-Hastings algorithm and produce sample from a specified target distribution. As we also wanted to emphasis the broader application to graph partitioning as our appreciation for these directions has grown.}

  There is, however, a growing interest in inquiring into the properties of a `typical' graph partition given certain partitioning criteria, particular in spatial clustering and Bayesian statistics (e.g. \cite{anderson2017spatial,balocchi2019crime}). Formally, one places a probability measure on the space of graph partitions and then investigates typical structures under this distribution. In practice partitions are sampled with techniques such such as Monte Carlo algorithms. 

Commonly, these methods invoke Markov Chain Monte Carlo methods in which one begins with some non-zero probability state (partition) and takes a random walk on the state space. 
A simple formulation of such a walk consists of choosing a single vertex on the boundary of a partition and proposing to exchange it to a neighboring partition 
\cite{macmillan2001redistricting,MattinglyVaughn2014,QuantifyingGerrymandering,Wu15,jcmReport,Chikina_Frieze_Pegden_2017,chikina2019separating}.  
Using ideas drawn from sampling spin glass models \cite{swendsen1987nonuniversal}, one group has examined the possibility of using percolation clusters to accelerate vertex exhange \cite{fifield2015}.

One common constraint is to demand balanced partitions, which is to say that the sums of node weights roughly match across partitions. Although the boundary exchange methods mentioned above are provably ergodic under soft constraints, the chains can have difficulties mixing due to large energetic barriers (though they still prove useful in some practical settings \cite{jcmReport}). As a potential remedy to these issues, a research group proposed a recombination (ReCom) algorithm which joins adjacent partitions, draws a spanning tree on the joined partitions, and generates a new partition by cutting the new spanning tree \cite{moonVa, deford2019redistricting,DeFord2018,DeFordDuchinPrivite}.

The ReCom method has demonstrated positive mixing properties on balance partitions on planar graphs, however it samples from an unknown invariant distribution on the space of partitions. \footnote{It is worth mentioning that there are a series of other alternative algorithms to ReCom that algorithmically generate random graph partitions if one is not concerned with knowing the sampling measure or is not interested in having the flexibility of specifying the sampled measure.  For example, the constructive algorithms of Chen and Rodden \cite{ChenRodden13} randomly merge adjacent nodes and the genetic algorithms of Cho \cite{Liu16} involve merging two partitions to create a new partition.  These procedures also do not have a known invariant measure nor are practical candidates for Metropolizing.} The algorithm has been shown to sample partitions from a measure that is `close' to sampling a space of balanced partitions weighted according to the product of the number of spanning trees that can be drawn within each partition.  However, it is unclear how to quantify this relationship generically. Additionally, it is computationally infeasible to use the original ReCom scheme as a proposal in a reversible Metropolis-Hastings scheme. Hence ReCom cannot be effectively  used to sample from a  specified measures of interest. This limits its usefulness in many applications which require an MCMC scheme which sample from a specified measure of interest.  This note shows how to extend the ideas in ReCom to make Metropolizing feasible. 

In the current work, we use a modified ReCom Markov chain as a proposal in order to develop a reversible Markov Chain Monte Carlo (MCMC) algorithm to sample a specified distribution on the space of balanced graph partitions. 
The global moves of ReCom promise faster mixing MCMC algorithms when used as the proposal chain, and the added (reversible) rejection step allows us to specify and alter the underlying measure being sampled. 
We alter the ReCom algorithm by extending graph partitions to track persistent spanning trees within each partition. We demonstrate that this expansion of the state space allows us to construct a reversible Markov chain that is able to completely propose redrawn pairs of adjacent partitions.

We demonstrate our algorithm in the context of redistricting. In this problem, we consider a geographical region which is to be split into several districts (i.e. graph partitions).  The region typically is comprised of a collection of small scale unit that will form the nodes of graph, and edges between nodes are pre-specified, often via geographic adjacency.  

In general, graph partition problems, either in the context of optimization or sampling are NP-hard or even NP-complete (e.g. see \cite{bulucc2016recent} and \cite{njatDedfordSolomon2019graphs}, respectively). From the sampling perspective, we are often not interested in recovering a distribution on partitions, but rather several observables of interest. This is analogous to sampling on molecular dynamic problems and recovering observables of interest such as chemical potentials. With this in mind, we present numerical evidence that the complexity of recovering robust statistics on low-dimensional observables of interest is significantly faster than attempting to recover the correct distribution on all partitions via sampling.

\begin{figure}
\centering
\subcaptionbox{Graph\label{sfig:overviewGraph}}{\includegraphics[width=0.27\linewidth, clip = true, trim = {4cm 5cm 3cm 4cm}]{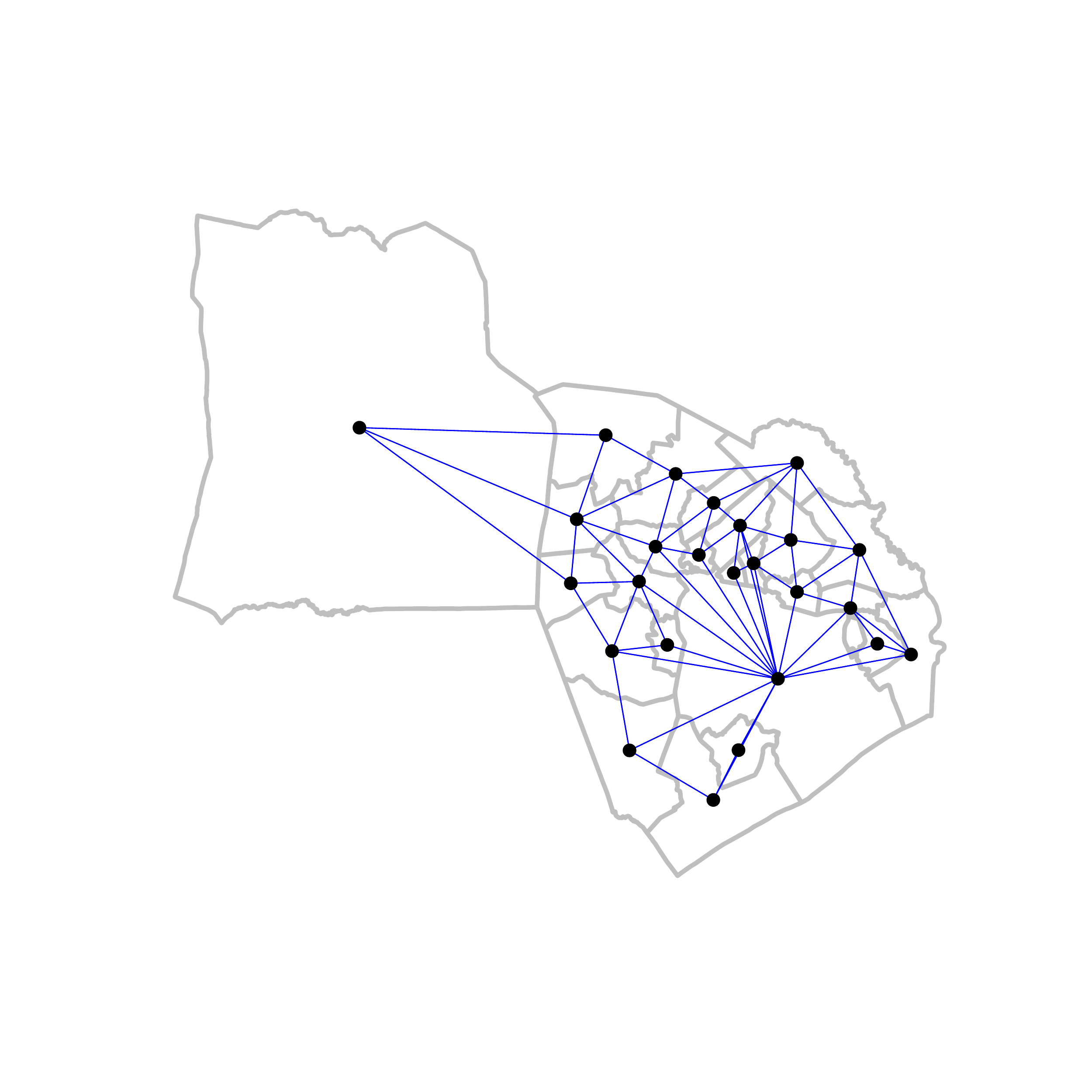}}\qquad
\subcaptionbox{Graph Partition\label{sfig:overviewGraphs}}{\includegraphics[width=0.27\linewidth, clip = true, trim = {4cm 5cm 3cm 4cm}]{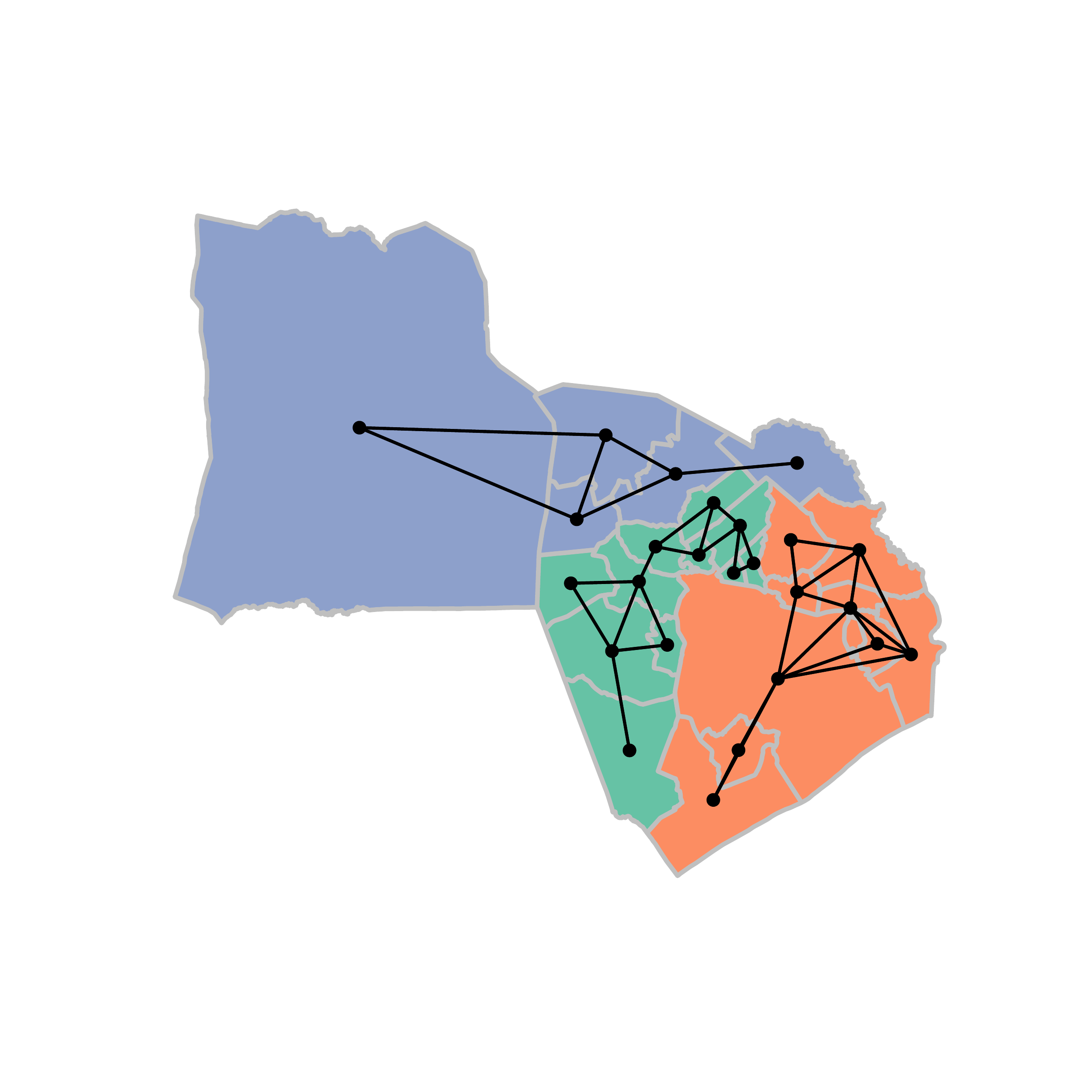}}\qquad
\subcaptionbox{Tree Partition (extended state)\label{sfig:overviewTrees}}{\includegraphics[width=0.27\linewidth, clip = true, trim = {4cm 5cm 3cm 4cm}]{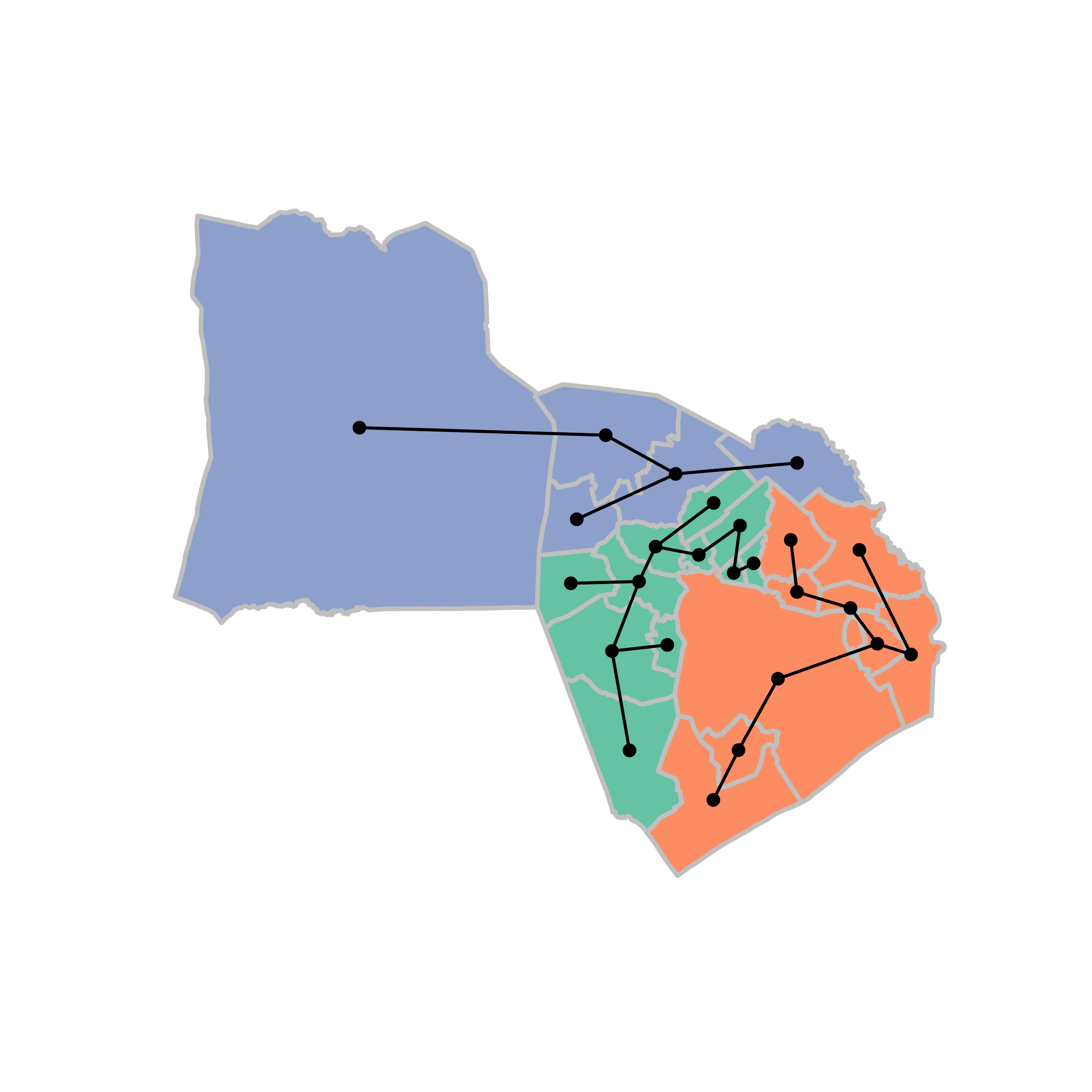}}

\subcaptionbox{Merge \& Sample New Tree on Merged Subgraph\label{sfig:overviewMerge}}{\includegraphics[width=0.27\linewidth, clip = true, trim = {4cm 5cm 3cm 4cm}]{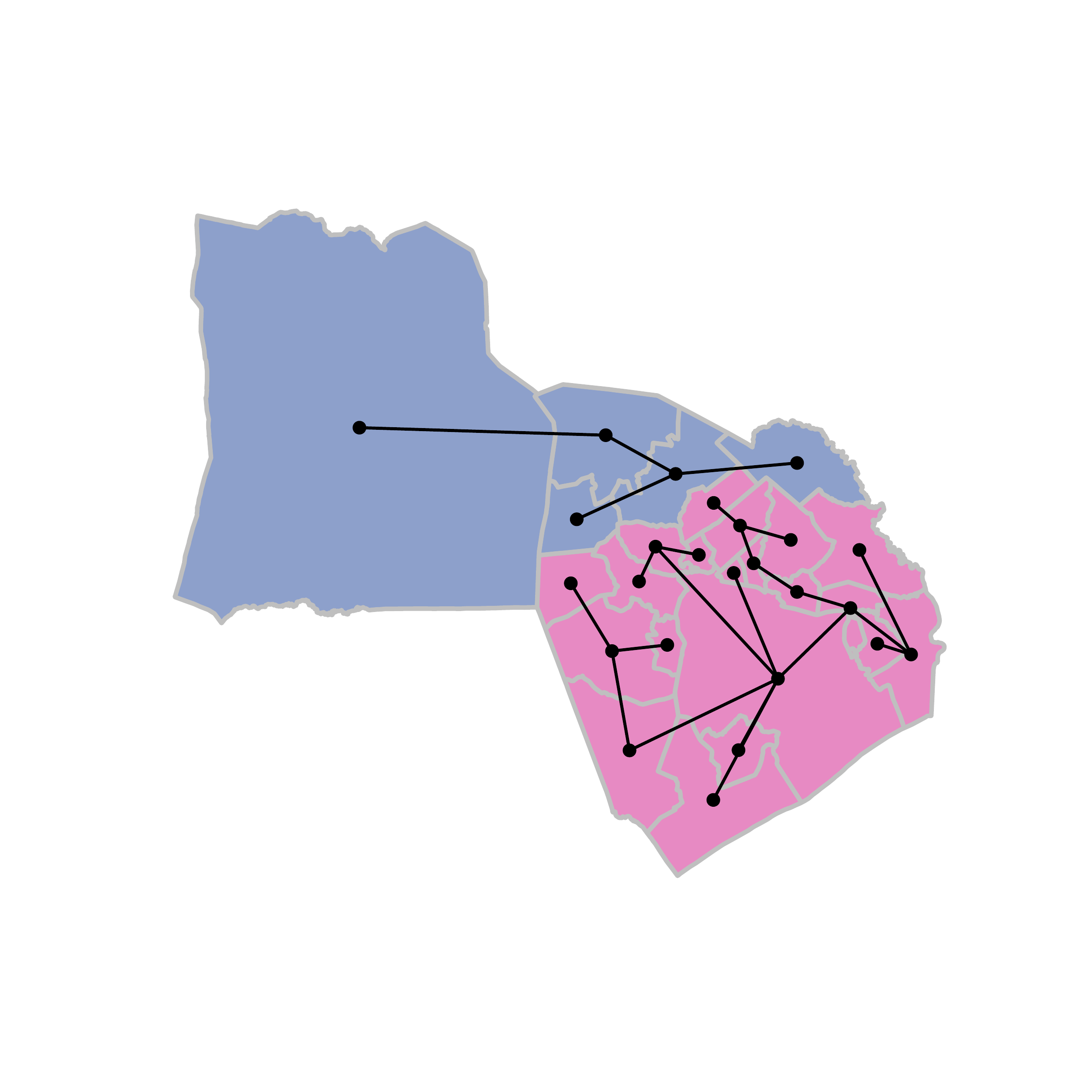}}\qquad
\subcaptionbox{Find Edge to Cut\label{sfig:overviewCut}}{\includegraphics[width=0.27\linewidth, clip = true, trim = {4cm 5cm 3cm 4cm}]{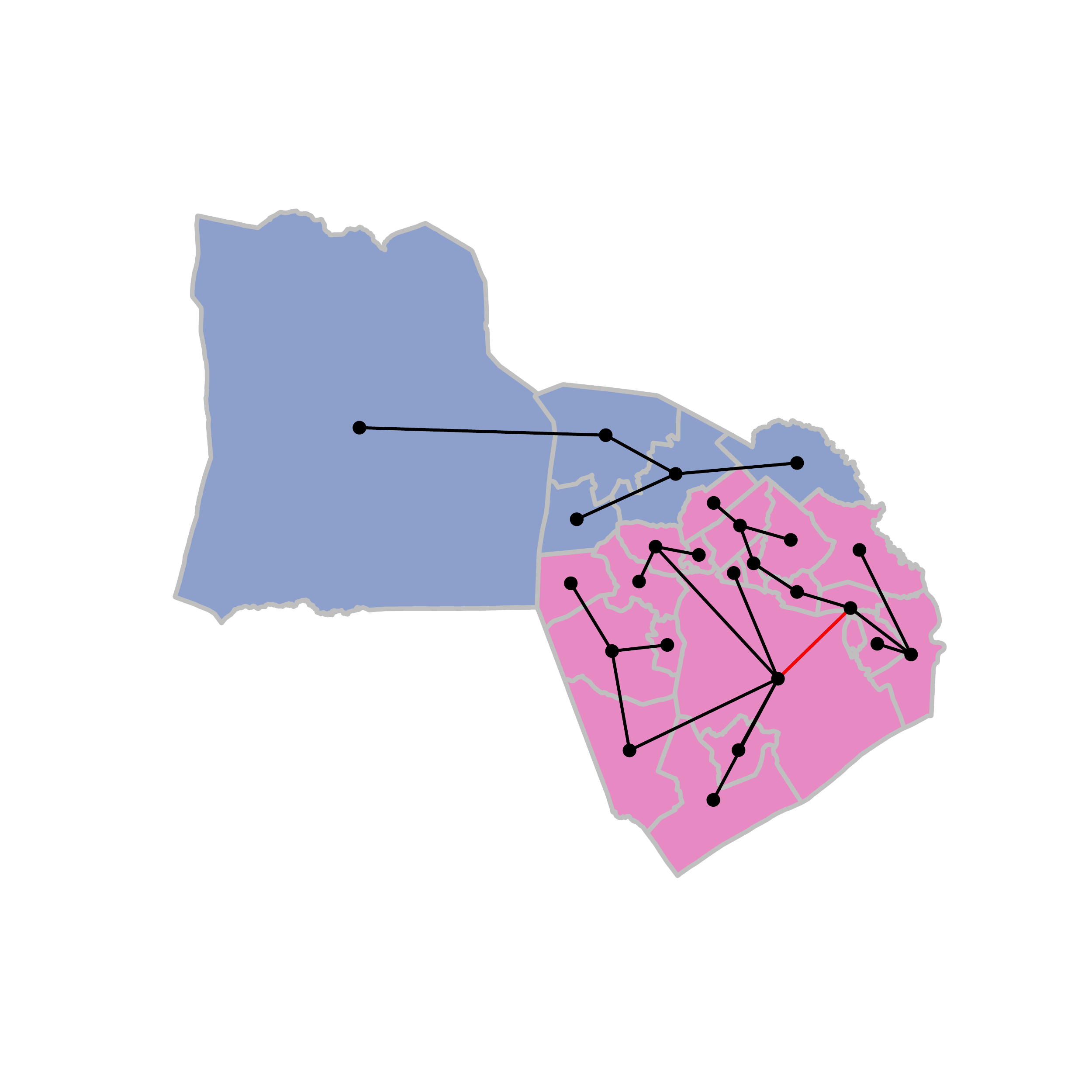}}\qquad
\subcaptionbox{Split into Two Trees\label{sfig:overviewSplit}}{\includegraphics[width=0.27\linewidth, clip = true, trim = {4cm 5cm 3cm 4cm}]{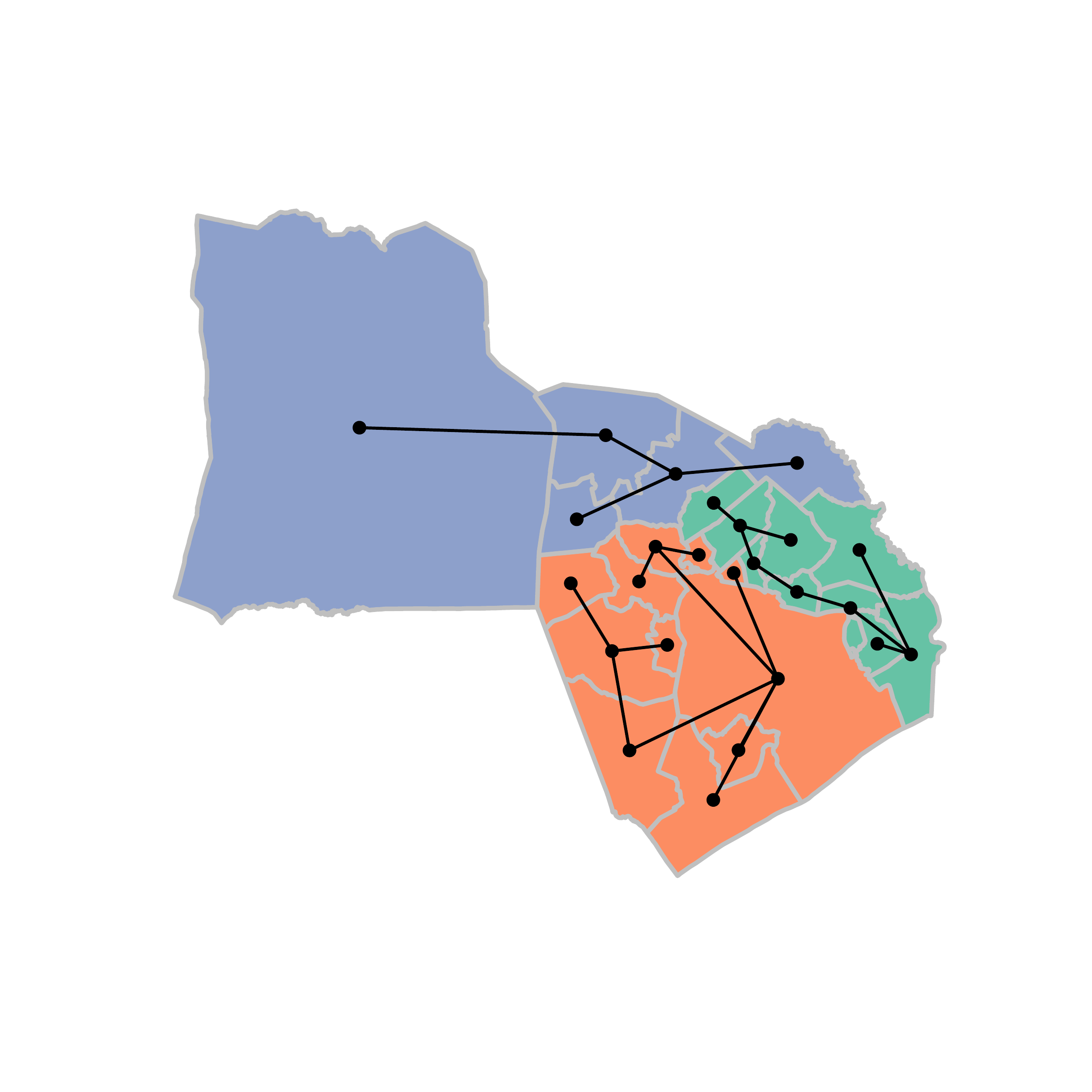}}
\caption{As an illustrative example, we consider a graph induced geometric regions (A). We then show how a partition of the graph induces subgraphs (B), how we extend the state to consider trees on the subgraphs (C), how a merge step might look (D), how edges are removed and how a subsequent split step might look (E, F).}
\label{fig:overview}
\end{figure}
\section{Informal Overview of the Metropolized Forest Recombination Algorithm}
Although our algorithm is one for generic graph partitioning, we illustrate it in the context of redistricting as a concrete example is useful when explaining the algorithm.
Typically, a region is partitioned by assigning small atomic geographic elements to a district or partition (see Figures~\ref{sfig:overviewGraph}~and~\ref{sfig:overviewGraphs}). At the highest level, the core of the algorithm, is similar to ReCom \cite{moonVa, deford2019redistricting,DeFord2018}. We pick two adjacent partitions, merge them into a single subgraph, and randomly redivide the merged subgraph into two new balanced partitions. This process is then repeated.


In ReCom, the new partition is constructed by first generating a random spanning tree on the merged subgraph (Figure~\ref{sfig:overviewMerge}). This spanning tree can then be used to efficiently divide the merged subgraph into two balance partitions.  In the redistricting problem, for example, balance would mean that each partition contain a population that are equal up to some tolerance (Figure~\ref{sfig:overviewCut}). One such partition is chosen randomly from those which are possible by cutting the given spanning tree in two (Figure~\ref{sfig:overviewSplit}).

We are interested in re-imagining ReCom so that it is practical as the random proposal in a Metropolis-Hastings MCMC algorithm. To employ Metropolis-Hastings, we need to be able to calculate the probability of proposing a particular move from an initial state and also the probability of proposing the reverse move from the proposed state back to the initial state. In ReCom, these calculations are impractical without further insight. 

To make the calculations feasible, we enlarge the state-space form graph parititons to spanning forests (Figure~\ref{sfig:overviewTrees}).  The extra information contained in the spanning forest representation greatly simplifies the calculation for the forward and reverse proposal probabilities (see Section~\ref{ssec:whylift}). \footnote{The methods in this work are not only method to make these calculations feasible; in particular, there is work in progress that adds self-loops in the Markov Chain of the ReCom scheme which allows one to determine the forward and backward proposal probabilities \cite{CannonConference2020}.}

\section{The Setting and Target Measure}
We now lay the groundwork we need to define the algorithm sketched in the previous section more formally.
We base our notation on that found in \cite{Bangia17,MattinglyVaughn2014} then expand upon it below. Let the graph $G$ have vertices $V$ and edges $E$. For example, each vertex may represent some sub-geographic region in a redistricting; in this context, edges are placed between vertices that are either rook, queen, or legally adjacent.\footnote{Rook adjacency means that the geographical boundary between two regions has non-zero length; queen adjacency means that the boundaries touch, but may do so at a point. At times two regions may not be geographically adjacent, but may be considered adjacent for legal purposes; for example, an island may still be considered adjacent to regions on a mainland for the purposes of making districts.}
Furthermore, in this context, we will be working with (mostly\footnote{At times, certain regions that represent a node may not be connected. If a node represents such a region, it is possible for the graph to be non-planar.}) planar graphs, however all of the ideas we will discuss are applicable to generic graphs.

We represent an $n$-partition on $G$ as a function $\xi:V \to \{1,2 \dots n\}$. Informally, $\xi(v)=i$ means $v$ is in the $i$th partition. Given a partition $\xi$, we will denote by 
$V_i(\xi) = \{v \in V \mid \xi(v) = i\}$ and $E_i(\xi) = \{(v,u) \in E \mid \xi(v) = \xi(u) = i\}$ respectively the set of vertices in the $i$th partition and the set of edges between vertices in the $i$th partition. We will define $\xi_i = (V_i(\xi), E_i(\xi))$ to be the subgraph induced by the $i$th partition.  

We will also sometimes associate extra data with the vertices and edges.  In the redistricting problem, this can include population, land area, and border length between vertices. The additional data is used to evaluate the districts on desired redistricting criteria, such as equal-population (i.e. balance between partitions) and geographical compactness. Of particular note, we define $\pop(v)$ to be the population of vertex $v$ and
\begin{align}
\pop(\xi_i) = \sum_{v\in V_i(\xi)} \pop(v).
\end{align}
to be the population of district $\xi_i$.  More generically, $\pop(\xi_i)$ may be thought of a weighted node sum on a generic partition.

In many settings, partitions must be simply connected and hence the state space is a subset of the set of $n$-partitions of the vertex set of a graph, where $n$ is the number of partitions. Using $\xi_i$ as above to represent the subgraph associated to the $i$th partition, we can recast $\xi$ as a map on vertices as the $n$ subgraphs defining the partitions. That is,
\[ \xi = \{\xi_1,\xi_2,\dots,\xi_n\}. \]
Since (up to the $n!$ equivalent labelings) there is a one-to-one correspondence between labeling functions $\xi$ and partitions into $n$ subgraphs  $\{\xi_1,\xi_2,\dots,\xi_n\}$ we will move between the two perspectives as convenient and consider $\xi$ to be both the labeling function and the partition.

We will see that in the context of Metropolizing our Forest ReCom proposal procedure it will be important to work on a space with more structure than the space of $n$-partitions. We choose to make our state space the set of $n$-\textit{tree partitions} of a graph; this is the space of forests consisting of $n$ disjoint trees whose union spans the vertices of the graph. We will use the term  \textit{spanning forest} interchangeably for such a collection of disjoint trees which span the graph. From this perspective the state space has elements of the form
\begin{align*}
 \allT = \{T_1,T_2,\cdots,T_n\},
\end{align*}
where each $T_i$ is a spanning tree on the subgraph $\xi_i$ with vertices $v_i=V_i(\xi)$ and edges $\varepsilon_i\subseteq E_i(\xi)$. The use of $n$-tree partitions of a graph rather than $n$-partitions only enlarges the state space. Thus, any distribution on the second can be represented on the first. However, we will see in Section~\ref{sec:samplingFromP} that this additional richness will allow us to build a fast and feasible algorithm for calculating proposal probabilities.
This extension is illustrated between Figure~\ref{sfig:overviewGraphs}~and~\ref{sfig:overviewTrees}. Henceforth we will consider our state to be a collection  disjoint spanning trees $\{T_i\}$ rather than a collection of disjoint graphs $\{\xi_i\}$. The one-to-one correspondence between partitions and labeling functions no longer holds, but since $\{T_i\}$ naturally induces $\{\xi_i\}$ (but not the converse), we will still consider the labeling function $\xi$ corresponding to $T=\{T_i\}$ and denote it by $\xi(T)$ or just $\xi$ if context makes the intent clear.

\subsection{The target measure on spanning forests}
We will now place the probability measure on this space spanning forest consisting of $n$ disjoint trees $\allT = \{T_1,T_2,\cdots,T_n\}$. We take our measure to be of the form
\begin{align}\label{eq:Pdef}
P(\allT) \propto e^{-\beta J(\xi(\allT))} \tau(\xi(\allT))^{-\gamma},
\end{align}
where $J$ is a score function that evaluates how ``good'' a graph partition $\xi$ is,\footnote{Lower scores are ``better'' in the sense that a partition in question performs better when considering  criteria included in the definition of $J$. Examples of criteria often incorporated in $J$ in the context of redistricting are the compactness of the districts, the deviation from the ideal population of each district and the degree to which administrative units (such as counties or cites) are fragmented.  } $\beta\in[0,1]$ and $\gamma\in[0,1]$ are tempering parameters used to change the importance of the factors $J(\xi)$ and $\tau(\xi)$ respectively, and
\begin{align}
\tau(\xi) = \prod_{i = 1}^n \tau(\xi_i),
\end{align}
where $\tau(\xi_i)$ is the total number of spanning trees on the graph $\xi_i$.\footnote{See the Supplemental Material for an efficient way to compute $\tau(G)$ for any graph $G$.} Since once the partitioning $\xi$ is fixed one can choose the spanning tree for each subgraph independently, $\tau(\xi)$ counts the total number of spanning forests on $\xi$ when viewed as a collection of disjoint graphs with the requirement that each partitioned subgraph is covered by a single spanning tree. In other words, $\tau(\xi)$ is the number of different states in our enlarged state space which correspond to the same partitions $\xi$. 

 The score function $J$ encodes a preference for maps with lower scores. It also encodes absolute constraints, which is to say maps that are strictly not allowed in the ensemble, by setting $J(\xi) = \infty$ on those maps. 
 For example, we may constrain the space if the partitions are not sufficiently balanced. 
 It is worth noting that the structure of the spanning forest $T$ does not explicitly enter the measure, as the measure only depends on the underlying partitioning, $\xi$. However, as already mentioned, we will see that considering our space to be $T$ rather than $\xi$ will be important in computing forward and reverse probabilities under Metropolis-Hastings.

 \subsection{The structure of the measure}
We now collect a number of observations about the structure of the measure $P$ and various limiting cases in $\gamma$ and $\beta$.
We will write $P(\allT; \beta = b, \gamma = g)$ for the probability of seeing the spanning forest $\allT$ in the distribution in \eqref{eq:Pdef} when $\beta=b$ and $\gamma=g$.

\subsubsection*{Uniform Measure on Spanning Forests.} When  $\gamma = 0$ and $\beta \rightarrow 0$, $P(\allT)$ converges to the uniform measure on the spanning forest of $n$ trees which satisfy the constraints described by the score function $J$; that is to say $J(\xi) < \infty$. If we were to use the convention that $0\times\infty=0$ in the exponent, when $\gamma=\beta=0$ the measure becomes
\begin{align}
P(\allT; \beta = 0, \gamma = 0) \propto 1,
\end{align}
which is to say we recover the uniform measure on the spanning forests, subject to no constraints. 

\subsubsection*{Uniform on All Graph Partitions} When $\gamma=1$, the distribution on graph partitions depends only on the factor involving $J$; that is, the probability of finding partition $\xi$ no longer depends on $\tau(\xi)$. To see this, note that
\begin{align}
P(\xi) &\propto \sum_{\allT \in ST(\xi)} P(\allT) = e^{-\beta J(\xi)} \frac{\tree(\xi)}{\tree(\xi)^\gamma},
\label{eqn:propxi}
\end{align}
where 
\begin{align}
ST(\xi) = ST(\xi_1)\times\cdots\times ST(\xi_n),
\end{align}
is the cartesian product of all spanning trees, $ST(\xi_i)$, of subgraph $\xi_i$. 

When $\gamma = 1$ and $\beta\rightarrow 0$, the measure becomes uniform on graph partitions subject to the absolute constraints given by the score function $J$. When $\gamma = 1$ and $\beta=0$ (as before, using the convention that $0\times\infty=0$), the measure is uniform on all graph partitions.

\subsubsection*{Intermediate Values of $\gamma$} As can be seen in equation~\eqref{eqn:propxi}, as $\gamma$ becomes smaller, we favor partitions that have a larger product of tree counts on the subgraphs.
In particular, when $\gamma = 0$ the chance of finding a partition, $\xi$, is proportional to the product of the number of spanning trees on the subgraphs in $\xi$.

For moderately sized graphs with a few hundred or thousand vertices, the number of spanning trees is extremely large. In fact, this number grows faster than exponentially with the number of vertices in the graph, assuming the graph has average degree larger than 2 \cite{greenhillAverageNumberSpanning2017}. This rapid growth may cause large disparities between the relative probabilities of different partitions, $\xi$, as this ratio will be proportional to the product of spanning tree ratios
\begin{align}
\frac{P(\allD)}{P(\allD')} \propto \frac{\tree(\allD)}{\tree(\allD')} \frac{\tree(\allD')^\gamma}{\tree(\allD)^\gamma} e^{-\beta[ J(\xi)- J(\xi')]}.
\label{eqn:treeratio}
\end{align}
When taking a random walk through the state space using the Metropolis-Hastings algorithm, proposed states will usually have either far fewer or far more trees than the prior state, and the acceptance probability will be dominated by this ratio. This issue may be alleviated by modifying $\gamma$ in the interval $[0,1]$. There is a trade-off: choosing $\gamma$ close to 0 leads to more similar probabilities and thus potentially better movement around the state space, but choosing $\gamma$ close to 1 leads to a distribution which is closer to uniform on the graph partitions rather than spanning forests.

\subsubsection*{Induced Measure on Partitions} 

It is instructive to pause and consider the relative structure of the measures on spanning forests and partitions. The following lemma shows that all spanning forests which correspond to a given partition are equally likely to be sampled. In other words, the measure conditioned on a given partition is uniform on the spanning trees which correspond to that partition.

\begin{lemma}
  If two spanning forests $T$ and $T'$ represent the same partition then their corresponding states have equal probability under the measure $P$. Given our notation, we may write that if $\xi(\allT)=\xi(\allT')$, then
  $P(\allT)=P(\allT')$.
\end{lemma}
\begin{proof}
  This follows from the fact that for a given spanning forest $T$, both the score function $J$ and the number of spanning trees $\tau$ only depend on the partition $\xi(T)$. Since these are the only occurrences of $T$ in the definition of $P$, the result follows.
\end{proof}

\section{Sampling From The Measure \texorpdfstring{$P$}{P}}
\label{sec:samplingFromP}

As already discussed, we will use a global merge and split algorithm to propose moves to the standard Metropolis-Hastings algorithm. Our Forest ReCom algorithm is not itself reversible, but the resulting Markov chain given by the Metropolis-Hastings algorithm, with Forest ReCom as a proposal, will be. Additionally, by Metropolizing, we are able to sample from a wide range of target measures. Although one can always in theory use Metropolis-Hastings to create a reversible chain from any proposal method able to sample from a specified target measure, Metropolis-Hastings will fail in practice if the rejection probabilities are too large or if calculating the necessary transition probabilities is computationally infeasible.
Previously similar algorithms described in \cite{moonVa, deford2019redistricting,DeFord2018}, failed to create reversible chains; 
ours does because it  manages to efficiently compute the forward and backwards proposal probabilities due to the extension of the state space.

In the next section, we review the Metropolis-Hastings algorithm in this setting. In Section~\ref{sec:MS}, we 
give a full description of our Metropolized Forest ReCom algorithm and many of the implementation details. We also explain what is gained computationally by working on the space of spanning forests rather than the space of partitions.

\subsection{Metropolis-Hastings algorithm}

To sample from the measure $P$ defined previously, we use the Metropolis-Hastings algorithm with our Forest ReCom  algorithm as the proposal method.
We will denote by $Q(T,T')$ the probability of starting from the spanning forest $T$ and proposing spanning forest $T'$ using the Forest ReCom algorithm. In other words, if the current state of chain is the spanning forest $T$, the measure $Q(T,\;\cdot\;)$ is the distribution of the next proposed move of the chain. Then, following the Metropolis-Hastings prescription, this move is accepted with a probability $A(T,T')$ defined by 
\begin{align}\label  {eq:A}
A(T, T') = \min\left(1, \frac{P(T')}{P(T)}\frac{Q(T', T)}{Q(T, T')}\right)=\min\left(1, e^{-\beta[ J(\xi')- J(\xi)]} \left[\frac{\tree(\xi)}{\tree(\xi')}\right]^\gamma  \frac{Q(T', T)}{Q(T, T')}\right),
\end{align}
and rejected with probability $1-A(T,T')$. If the step is accepted, the next state is the proposed state; if the step is rejected, the next state does not change.
One can prove under relatively mild considerations that this process will converge to sampling from the measure $P$ if it is run for sufficiently many steps.

\subsection{The Forest ReCom Proposal}\label{sec:MS}
We now describe the Forest ReCom proposal Markov Chain $Q$ introduced in previous section. As already mentioned, this
  algorithm is specifically designed to have both forward and backward transition probabilities which can be efficiently computed. From \eqref{eq:A}, we see that this is critical if it is to be used in the Metropolis-Hastings algorithm as a proposal.

\label{sec:mergeSplitQ}
The proposal step shares many features with the ReCom algorithm \cite{deford2019recombination}.  The differences come in because we are evolving spanning forests rather than graph partitions. We provide a brief outline of this proposal, assuming that the current state of the chain is the spanning forest $T$. Our goal is to produce a new spanning forest, $T'$, which corresponds to merging two adjacent spanning trees in $T$, then redividing the merged tree into two new spanning trees which satisfy constraints given by $J$, and  then calculate $Q(T,T')$ and $Q(T',T)$. In this sub-section, we provide a high-level outline of this procedure.

Given spanning forest $T=(T_1,\cdots,T_n)$, we
\begin{enumerate}
    \item Choose two trees $T_i$ and $T_j$ from $T$ which correspond to adjacent partitions.
    \item Draw a new spanning tree $T_{ij}'$ uniformly at random on the subgraph $\xi_{i,j}$ induced by the union of vertices in $T_i$ and $T_j$ and the edges connecting these vertices. In other words, this induced graph is $\xi_{ij}= (V_{ij}, E_{ij})$ where $V_{ij}(\xi) = \{v \in V \mid \xi(v) \in  \{i,j\} \}$ and $E_{ij}(\xi) = \{(v,u) \in E \mid \xi(v), \xi(u) \in \{i,j\} \}$.
    \item Determine the edges of the newly-drawn tree $T_{ij}'$ such that, once removed, they would split the spanning tree into two trees that each comply with some subset of the constraints. If there are no such edges, remain in the current state and the proposal is finished; otherwise continue.
    \item Select one such edge and remove it from the new spanning tree $T_{ij}'$, leaving two new trees $T_i'$ and $T_j'$.
    \end{enumerate}
    
    In detailing the implementation of these steps, there are a number of ways that the first step may be carried out. For example, we might chose uniformly from all pairs of adjacent partitions; alternatively we can weight the choice by some property of the shared boundary between partitions such as the sum of edge weights; or we can account for the the values of the score function $J$. 

    The second step is achieved by Wilson's algorithm which employs loop-erased random walks. For the third step, the most pertinent constraint is equal population, so the third step involves a simple depth-first search along the tree with exit criteria based on the remaining population within a search branch. Choosing the specific edge to cut in the fourth step may be done uniformly or with a weighted distribution that might, for example, favor more equal populations.


    To make the algorithm reversible through Metropolis-Hastings, we must compute the probability of proposing state $T'$ from state $T$, along with computing the reverse probability of proposing state $T$ from state $T'$.
    Once the pair of spanning trees  $T_i$ and $T_j$, to be are merged and then split, is chosen, the remainder of the algorithm is summarized by the set of mappings show in equation \eqref{eq:algFlow}. They summarize steps 2--4 above. The annotations will help to explain why the choice of a forest of spanning trees as states space and the particular structure of \eqref{eq:algFlow} is important for calculating the probability of forward and backward probabilities.
    \begin{align}
      \label{eq:algFlow}
      \{T_i, T_j\} \xrightarrow[\text{deterministic}]{\text{many-to-one}}  \xi_{ij} \xrightarrow[\text{random}]{\text{one-to-many}}  T_{ij}' \xrightarrow[\text{random}]{\text{one-to-a-few}}   \{T_i', T_j'\} 
    \end{align}
There are some indications in the above mappings that suggest calculating the forward and backward probabilities should be tractable. The initial mapping is deterministic and all the random choices to come only depend on $\xi_{ij}$. Though the next step is a one-to-many random map, we will choose it to be uniform on a set whose size we can calculate, and, from which, we can draw uniformly. Since the next map is onto a relatively small set, it will be possible   calculate the forward probabilities and produce a random draw. Because in this step the forward possibilities are limited, we will see that identifying the backward possibilities will also be tractable. 

\subsection*{Calculating the forward and backward proposal probabilities, \texorpdfstring{$Q(T,T')$}{Q(T,T')} and \texorpdfstring{$Q(T',T)$}{Q(T',T)}}
We begin by calculating the chance that the merging of $\{T_i, T_j\}$ and subsequent splitting produces the replacement spanning trees $\{T_i', T_j'\}$.  To do this, we must examine all possible spanning trees on the induced graph of $\xi_{ij}$ that could have been drawn in step 2 and then cut in step 4 to result in $T_i'$ and $T_j'$. We must then sum the probability that we found $T_i'$ and $T_j'$ across all such choices to compute the probability of proposing the new state.

The set of all possible spanning trees that could result in $T_i'$ and $T_j'$ is simply the trees defined by $T_i'\cup T_j'$ along with each edge in $G$ connecting the two graphs (see Figure~\ref{sfig:calcPropEdges}). The edges in $G$ that connect the spanning trees are $\{(v,u)\in E \mid \xi'(u) = i, \xi'(v) = j \}$ and we will denote this set as $E(T_i', T_j')$. 
Together with the new trees $T_i'$ and $T_j'$, each edge $e\in E(T_i', T_j')$ induces a spanning tree on the induced graph $\xi_{ij}$, which we will denote $T_{(T_i', T_j', e)}$. This new spanning tree is one of the trees that could have been drawn in step 2 and cut in step 4 to yield $T_i'$ and $T_j'$.

Let the probability from step 4 that we cut the tree $T_{(T_i', T_j', e)}$ at edge $e'$ be $P_{cut}(e' \mid T_{(T_i', T_j', e)})$. Note that the probability that we cut $T_{(T_i', T_j', e)}$ into $T_i'$ and $T_j'$ is $P_{cut}(e \mid T_{(T_i', T_j', e)})$ (see Figure~\ref{sfig:calcPropEdgeCuts}). Finally, note that the probability of drawing each of the spanning trees induced by some edge in $E(T_i', T_j')$ is simply $1/\tau(\xi_{ij})$.

Putting this all together, we now find that the probability of the proposing $T_i'$ and $T_j'$ from $T_i$ and $T_j$ is
\begin{align}
q(\{T_i, T_j\}, \{T_i', T_j'\}) = \frac{1}{\tau(\xi_{ij})} \sum_{e\in E(T_i', T_j')} P_{cut}(e \mid T_{(T_i', T_j', e)}).
\label{eqn:twoDistProposal}
\end{align}

This calculation assumes that we have already chosen $T_i$ and $T_j$. Given the spanning forest $T$ let us denote by $p(\{i,j\} \mid \allT)$ the probability 
from step one of picking the pair of adjacent spanning trees $T_i$ and $T_j$ to merge. This probability is simple to calculate for most reasonable choices of how to perform step 1. Then the probability of proposing state $T'$ from state $T$ is given as
\begin{align}\label{eq:MergeSplitQ}
Q(\allT, \allT') = p(\{i,j\} \mid \allT) q(\{T_i, T_j\}, \{T_i', T_j'\}).
\end{align}
where $q(\{T_i, T_j\}, \{T_i', T_j'\})$ is the transition calculated in \eqref{eqn:twoDistProposal}.

\begin{figure}
\centering
\subcaptionbox{Edges $E(T_i',T_j')$ (dashed) that complete a merged tree \label{sfig:calcPropEdges}}{\includegraphics[width=0.4\linewidth, clip = true, trim = {4cm 5cm 3cm 4cm}]{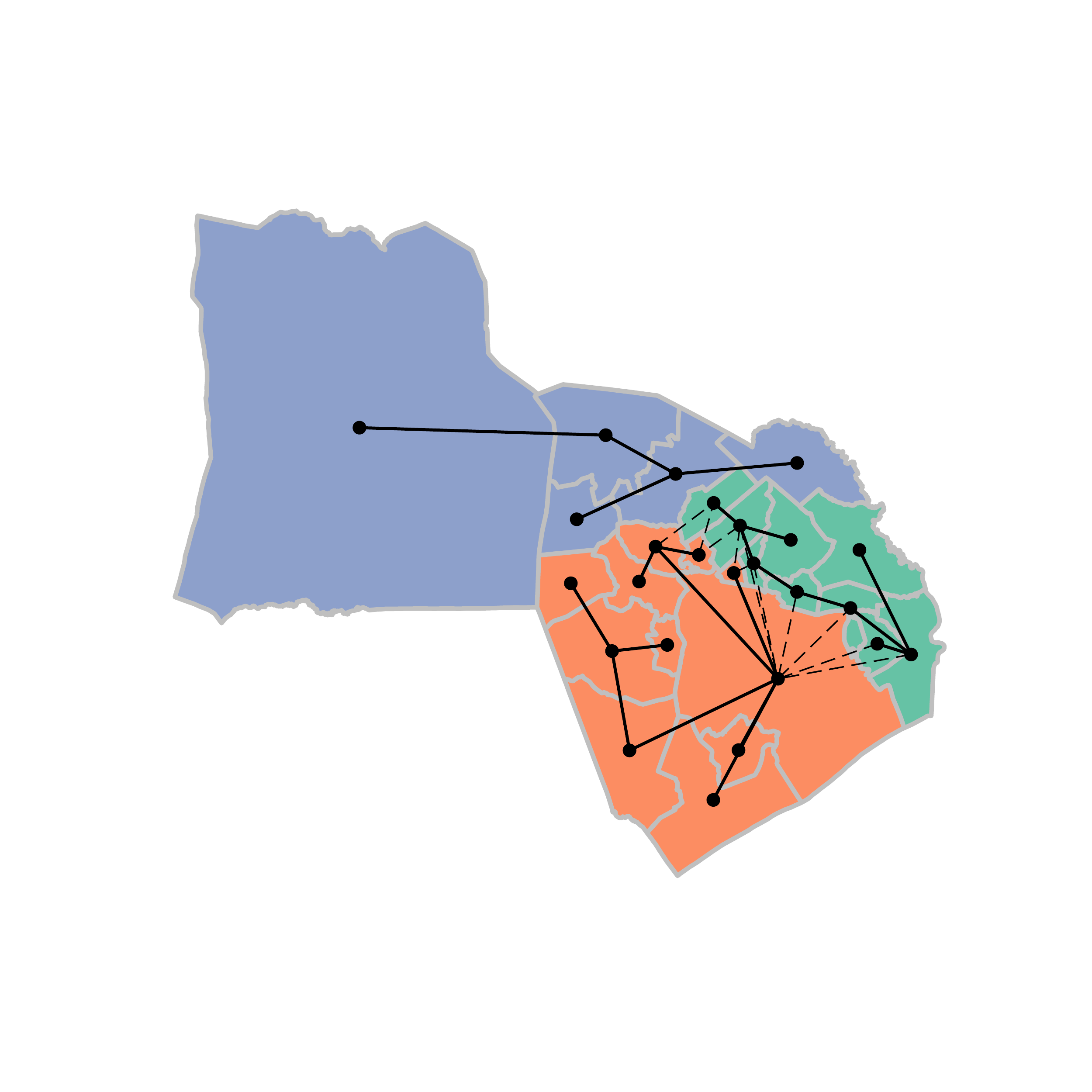}}\qquad
\subcaptionbox{After selecting one edge from (A), we highlight all edges that could have been cut to provide a new balanced partition \label{sfig:calcPropEdgeCuts}}{\includegraphics[width=0.4\linewidth, clip = true, trim = {4cm 5cm 3cm 4cm}]{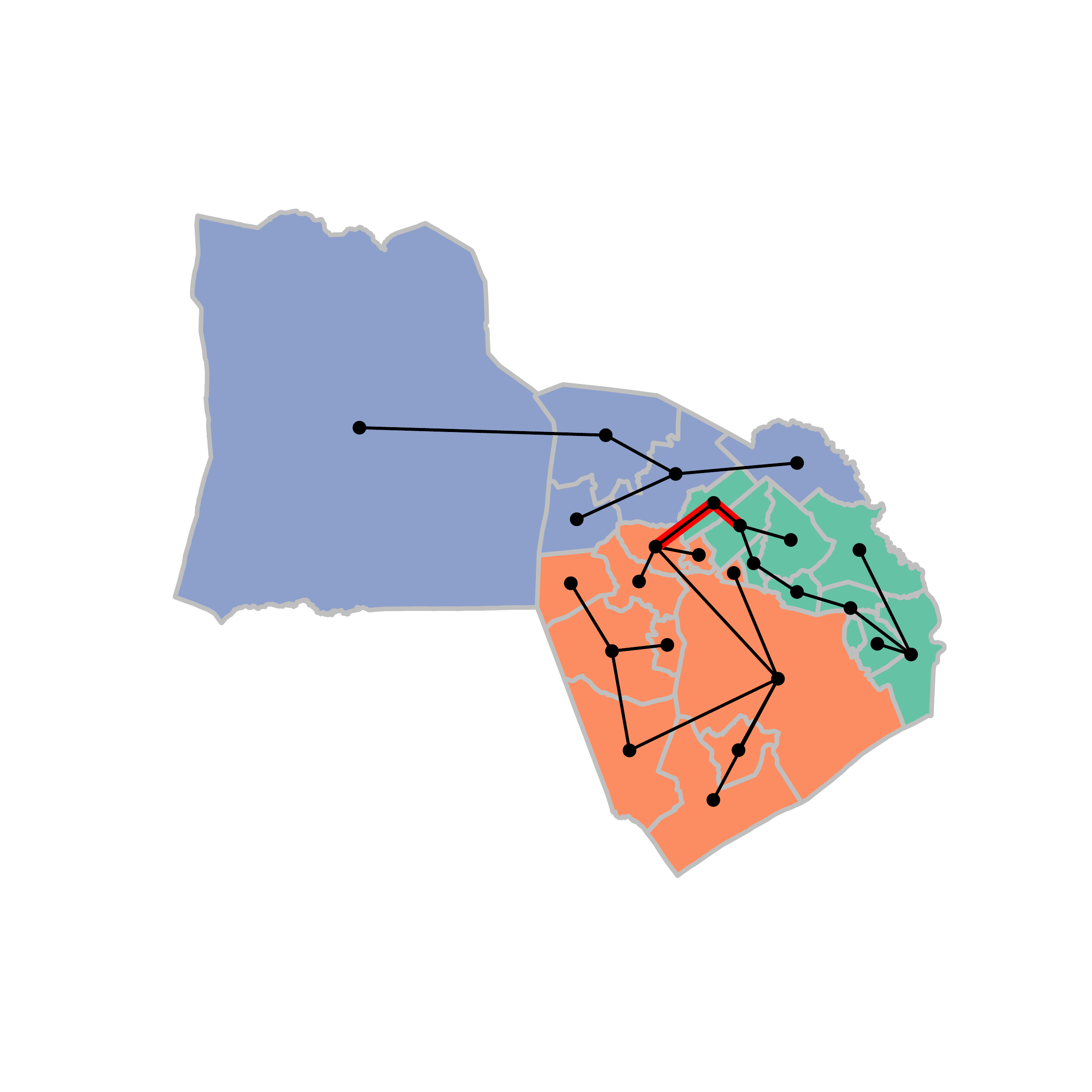}}\
\caption{
When calculating the proposal probability of drawing $T_i'$ (orange) and $T_j'$ (green), we must examine all edges that could have made a spanning tree on the joined space (A). For each edge we must examine all edges that could have been cut and then compute the probability that we cut the edge that leads to the observed partition. In (B) we show the choice of one of the conflicted edges in $e\in E(T_i', T_j')$ and mark it along with a second edge in red; both of these edges could have been cut with probability $P_{cut}(e'\mid T_{(T_i', T_j',e)})$, where $e'$ is either $e$ or the second edge highlighted in red. 
}
\label{fig:calcProp}
\end{figure}

\subsection*{Properties of the Forest ReCom Proposal \texorpdfstring{$Q$}{Q}} 
First observe that when the Forest ReCom proposal from \eqref{eq:MergeSplitQ} is inserted into the formula for the acceptance probability from \eqref{eq:A}, we obtain
 \begin{align}\label{AQ}
    A(T, T') = \min\left(1, e^{-\beta(J(\xi')-J(\xi))}\left[\frac{\tau(\xi)}{\tau(\xi')}\right]^\gamma\frac{p(\{i,j\} \mid \allT')} {p(\{i,j\} \mid \allT)}\frac{q(\{T_i', T_j'\}, \{T_i, T_j\})}{q(\{T_i, T_j\}, \{T_i', T_j'\})}\right). 
 \end{align}
From this we see that if $\gamma=0$, one does not need to calculate the $\tau(\xi)$ and $\tau(\xi')$ factors, which reduces the computational costs. Furthermore, as mentioned before, the difference between $\tau(\xi)$ and $\tau(\xi')$ is a principle reason for a low acceptance rate when $\gamma$ is closer to 1.
 
Inserting \eqref{eqn:twoDistProposal} into \eqref{AQ},  we see that the ratios of the Forest ReCom proposal probabilities can be written as
 \begin{align}
   \label{eq:qq}
   \frac{q(\{T_i', T_j'\}, \{T_i, T_j\})}{q(\{T_i, T_j\}, \{T_i', T_j'\})}= \frac{\sum_{e\in E(T_i, T_j)} P_{cut}(e \mid T_{(T_i, T_j, e)})}{\sum_{e\in E(T_i', T_j')} P_{cut}(e \mid T_{(T_i', T_j', e)})}\,.
 \end{align}
 In particular, we see that the $\tau(\xi_{ij})$ factors in each $q$ expression cancel and hence need not be computed.

To better under the structure of  this ratio, notice that when there is only a single edge that could possibly be cut for any spanning tree induced by the edges $E(T_i, T_j)$, we may write the proposal ratio as
 \begin{align}
   \frac{q(\{T_i', T_j'\}, \{T_i, T_j\})}{q(\{T_i, T_j\}, \{T_i', T_j'\})}= \frac{|E_{ij}(T_i, T_j)|}{|E_{ij}(T_i', T_j')|}\,.
 \end{align}
Hence in this case we see that the ratio of probabilities is equal to the ratio of the number of edges which connect the two trees, or, in other words, the graph-theoretic length of the boundary. With this calculation in mind it is reasonable to define an effective boundary between $T_i$ and $T_j$ by
\begin{align}
  \partial(T_i,T_j) = \sum_{e\in E(T_i, T_j)} P_{cut}(e | T_{(T_i, T_j, e)})\,.
\end{align}
Returning to general $\gamma$ and with this notation, the acceptance ratio becomes 
\begin{align}\label{eq:Asimplified}
  A(T, T') = \min\left( 1, e^{-\beta [ J(\xi') - J(\xi)] } \frac{p(\{i,j\}\mid T')}{p(\{i,j\}\mid T)} \frac{\partial(T_i,T_j)}{\partial(T_i',T_j')}\left[\frac{\tree(T_i) \tree(T_j)}{\tree(T_i') \tree(T_j')} \right]^\gamma\right),
\end{align}
where we have used the fact that the spanning forests $T$ and $T'$ only differ in the $i$th and $j$th trees so
\begin{align*}
  \frac{\tree(T)}{\tree(T')}= \frac{\tree(T_i)  \tree(T_j)}{\tree(T_i')  \tree(T_j')}.
\end{align*}

As mentioned previously, the ratio between spanning tree products, $\tree$, may be large between partitions, $\xi$ and $\xi'$. This disparity is eliminated when $\gamma = 0$, however setting $\gamma = 0$ favors sampling partitions with higher values of $\tree$. The parameter $\gamma$ is presented as a smoothly varying parameter because it demonstrates how one may use a tempering (e.g. simulated or parallel tempering) scheme to vary $\gamma$ from 0 to 1 across multiple chains in an extended product measure. 
We have chosen a form of the measure which essentially depends only on the partition $\xi(T)$ induced by the spanning forest $T$. 
Additionally, when $\gamma = 0$ there is no need to compute the number of spanning trees on the new partitions, as the products $\tree$ do not explicitly appear in the measure nor the proposal ratio. 

\subsection{Why lift from partitions to tree partitions?}
\label{ssec:whylift} 

We remark now on why it was useful to expand our state space on the space of all spanning forests.
Consider the merge split algorithm where given a partition $\xi=\{\xi_1,\cdots,\xi_n\}$, we produce a new partition $\xi'=\{\xi_1',\cdots,\xi_n'\}$. $\xi'$ is the same as $\xi$ except that two adjacent elements of the partition have been merged and then split in two to create two new elements. This is essentially the Forest ReCom proposal described in Section~\ref{sec:mergeSplitQ} and given by the probability distribution $Q$. Our algorithm begins by erasing the initial spanning trees on the two partitions chosen, merging the vertices, then drawing a new spanning tree on the induced graph. We could, however, also view the algorithm as moving between pairs of partition elements in which we draw a spanning tree on the merged induced graph and precede exactly as in Section~\ref{sec:mergeSplitQ} only without knowledge of the extended state space. Hence the two perspectives only differ in that one takes a spanning forest and transitions to a new spanning forest, whereas the second perspective takes a partition and transitions to a new partition.

The difference comes when one tries to calculate the probability $Q(\xi,\xi')$. We have already seen that in the spanning forest space $Q(T,T')$ is tractable. However, we now explain why, to the best of our understanding, calculating $Q(\xi,\xi')$ is much more difficult.

Looking back at \eqref{eqn:twoDistProposal}, we begin by remarking that the transition kernel 
$q(\{T_i, T_j\}, \{T_i', T_j'\})$ only depends on $\{T_i, T_j\}$ through the union of their vertex set which we have denoted $\xi_{ij}$. Hence we can equally view $q$ as a transition from the partitions $\{\xi_i, \xi_j\}$ to the spanning trees to the  $\{T_i', T_j'\}$ denoted by $q(\{\xi_i.\xi_k\}, \{T_i', T_j'\})$. This is simply because both $\{\xi_i,\xi_k\}$ and $\{T_i, T_j\}$ determine $\xi_{ij}$  and hence can be use as input to calculate the probability. (This last fact is evident from \eqref{eq:algFlow} and from \eqref{eq:algFlowPartitions} given below.
With this observation, we have that  
\begin{align}
  Q(\{\xi_i, \xi_j\}, \{\xi_i', \xi_j'\}) &= \frac{1}{\tau(\xi_{ij})}\sum_{T_i' \in ST(\xi_i')} \sum_{T_j' \in ST(\xi_j')}\sum_{e\in E(T_i', T_j')} P_{cut}(e \mid T_{(T_i', T_j', e)})\notag \\
  &= \sum_{T_i' \in ST(\xi_i')} \sum_{T_j' \in ST(\xi_j')} q(\{\xi_i.\xi_j\}, \{T_i', T_j'\}) \label{eq:Qxi}
\end{align}
where we have used the same notation as in Section~\ref{sec:mergeSplitQ}.
We now see that if we want to calculate $Q(\{\xi_i, \xi_j\}, \{\xi_i', \xi_j'\})$ on needs to calculate  $|ST(\xi_i')|\times |ST(\xi_j')|$ transition probabilities of the form $q(\{\xi_i.\xi_j\}, \{T_i', T_j'\})$ instead of the one needed for $q(\{T_i, T_j\}, \{T_i', T_j'\})$. This saves considerable computational effort.

Another perspective is that computing $Q(\{\xi_i, \xi_j\}, \{\xi_i', \xi_j'\})$ involves the one-to-many map of taking a partition into all possible spanning trees followed by the many-to-one map of erasing the spanning tree to obtain the partition. This is expensive. In  $Q(\{T_i, T_j\}, \{T_i', T_j'\})$ the order is reversed. The many-to-one operation of erasing the trees to arrive at a partition comes first followed by the one-to-many operation of drawing the trees. This computation is relatively easy.
This is summarized in diagram in equation \eqref{eq:algFlowPartitions} which should be compared with  equation \eqref{eq:algFlow}.
\begin{align}
      \label{eq:algFlowPartitions}
      \{\xi_i, \xi_j\} \xrightarrow[\text{deterministic}]{\text{many-to-one}}  \xi_{ij} \xrightarrow[\text{random}]{\text{one-to-many}}  T_{ij}' \xrightarrow[\text{random}]{\text{one-to-a-few}}   \{T_i', T_j'\} \xrightarrow[\text{random}]{\text{many-to-one}}   \{\xi_i', \xi_j'\} 
    \end{align}
When comparing with  \eqref{eq:algFlow}, we see that the first maps are essentially equivalent basically encapsulating the already mentioned fact that all of the random choices only depend on the initial state though $\xi_{ij}$. The next two random maps are the same as for the algorithm on spanning trees; and hence, are relatively easily to draw from and compute the backwards and forwards probabilities. The complication comes from the last mapping. It is many-to-one. This means that there are many tree pairs  $\{T_i', T_j'\}$ in the pre-image of a single partition pair  $\{\xi_i', \xi_j'\}$. This is represented by the two outer sums in \eqref{eq:Qxi} which combined are over $|ST(\xi_i')|\times |ST(\xi_j')|$ terms.

\begin{remark}
Up until this point, we have insisted that the state space be defined as a spanning forest. However, it is possible to use Metropolized Forest ReCom to sample from the space of partitions while preserving most of the computational improvements.  

To see this, observe that each spanning tree within each partition is uniform conditioned on that partition, thus conditioned on each partition, each tree in the partition is equally likely. Due to this, we can project the state space down the the space of partitions without losing critical information. We then modify the algorithm as follows: When we merge and split on adjacent partitions we can draw new random trees on each partition. There is a small amount of wasted effort on redrawing the trees within current partitions; Metropolized Forest ReCom does not waste this effort since it holds the trees within the state.  These trees are then used when calculating the reverse probability.  When we either accept or reject, the old or new trees can be deleted and we can project back down to the space of partitions.
\end{remark}
     
\section{Implementation of Forest ReCom Proposal \texorpdfstring{$Q$}{Q}} 
\label{sec:implementation}

We detail the implementations of steps within the proposal. 
There are many choices when picking the adjacent pair probability pair $P(\{i,j\}|T)$ -- one may uniformly choose amongst adjacent partitions or weight the choice by the sum of edge weights (e.g. shared border length for redistricting), the shared number of conflicted edges, or some heuristic of the acceptance probability. In the current work we make this choice in two different ways. In one method, we uniformly pick a random district and then uniformly pick a random district neighbor, so that 
\begin{align}
P(\{i,j\}|T) &= P(\{i,j\}|\xi(T))\nonumber\\
& = P(i | \xi) P(j | i, \xi) + P(j | \xi) P(i | j, \xi)\nonumber\\
& = \frac{1}{n}\left(\frac{1}{N_i(\xi)} + \frac{1}{N_j(\xi)}\right),
\label{eq:adjdist}
\end{align}
where $n$ is the number of partitions, and $N_i(\xi(T))$ is the number of partitions neighboring partition $\xi_i$. In the second method, we randomly select an edge that is not confined to a single partition, meaning that we preferentially choose partition pairs with more shared edges. Formally,
\begin{align}
P(\{i,j\}|T) = P(\{i,j\}|\xi(T)) = \frac{\partial_e (\xi_i, \xi_j)}{\displaystyle{\sum_{\ell<k}} \partial_e (\xi_\ell, \xi_k)},
\label{eq:bdrydist}
\end{align}
where $\partial_e (\xi_\ell, \xi_k)$ is the number of edges on the graph that have one node in partition $\xi_\ell$ and another in partition $\xi_k$. We denote the former method as ``uniform partition-uniform neighbor'' and the latter as ``boundary weighted.''

As already mentioned, uniform spanning trees will be drawn using Wilson's algorithm. There are several implementations of this algorithm (for details detail see the Supplemental Material). If $\gamma \neq 0$ we must compute the number of spanning trees on each subgraph induced by $\xi_i$; this is accomplished via Kirchoff's theorem and is, algorithmically, the slowest step of the algorithm (see the Supplemental Material).

When choosing what edge to cut, we must specify the probability of cutting edge $e$ given tree $T_{ij}$, $P_{cut}(e | T_{ij})$. Perhaps the simplest implementation is to uniformly choose an edge from the set of edges such that the the cut leads to two trees with populations within the constraints set out by $J$. Let $E_c(T)$ denote the edges, such that if a single edge was removed from $T$, the remaining two trees would have population within the specified constraints. Then \begin{align}
  P_{cut}(e | T) = 
  \begin{cases} 
    |E_c(T)|^{-1} & e \in E_c(T),\\
    0 & \text{ otherwise.}
  \end{cases}
\end{align}
We adopt this approach in the present work. We detail how we determine $E_c(T)$ both for the initial cut, and when computing the proposal probability given in equation~\eqref{eqn:twoDistProposal} in the Supplemental Material.

\section{Results}

\subsection{The Duplin-Onslow County Cluster} We implement and test the merge split algorithm by considering the graph generated by a ``county-cluster'' for the state House in North Carolina made up of Duplin and Onslow counties and comprising three state house legislative districts. In this case, edges of the graph are determined based on geographic adjacency and `district' will be used interchangeably with a partition, $\xi_i$ and a `districting plan' will be used interchangeably with a partition $\xi$. 

The two county, three district problem we are considering is a realistic one: In North Carolina, state legislative districts are first divided into county-clusters comprising a certain number of districts; each county-cluster forms a unique and independent region to redistrict.\footnote{See \cite{carter2019optimal} for more discussions and background on county clusterings in redistricting North Carolina}  The Duplin-Onslow cluster is comprised of three state House legislative districts. In constructing the graph to be partitioned, we make one more problem specific consideration which is that Onslow county has fewer people than a whole district and therefore it must be kept in tact. The Duplin-Onslow cluster is the example we have shown in Figures~\ref{fig:overview}~and~\ref{fig:calcProp}. We show the cluster within the state of North Carolina in Figure~\ref{fig:DOinNC}. 

\begin{figure}
\centering
\includegraphics[width = 0.8\linewidth, clip = true, trim = 5cm 8cm 5cm 7cm]{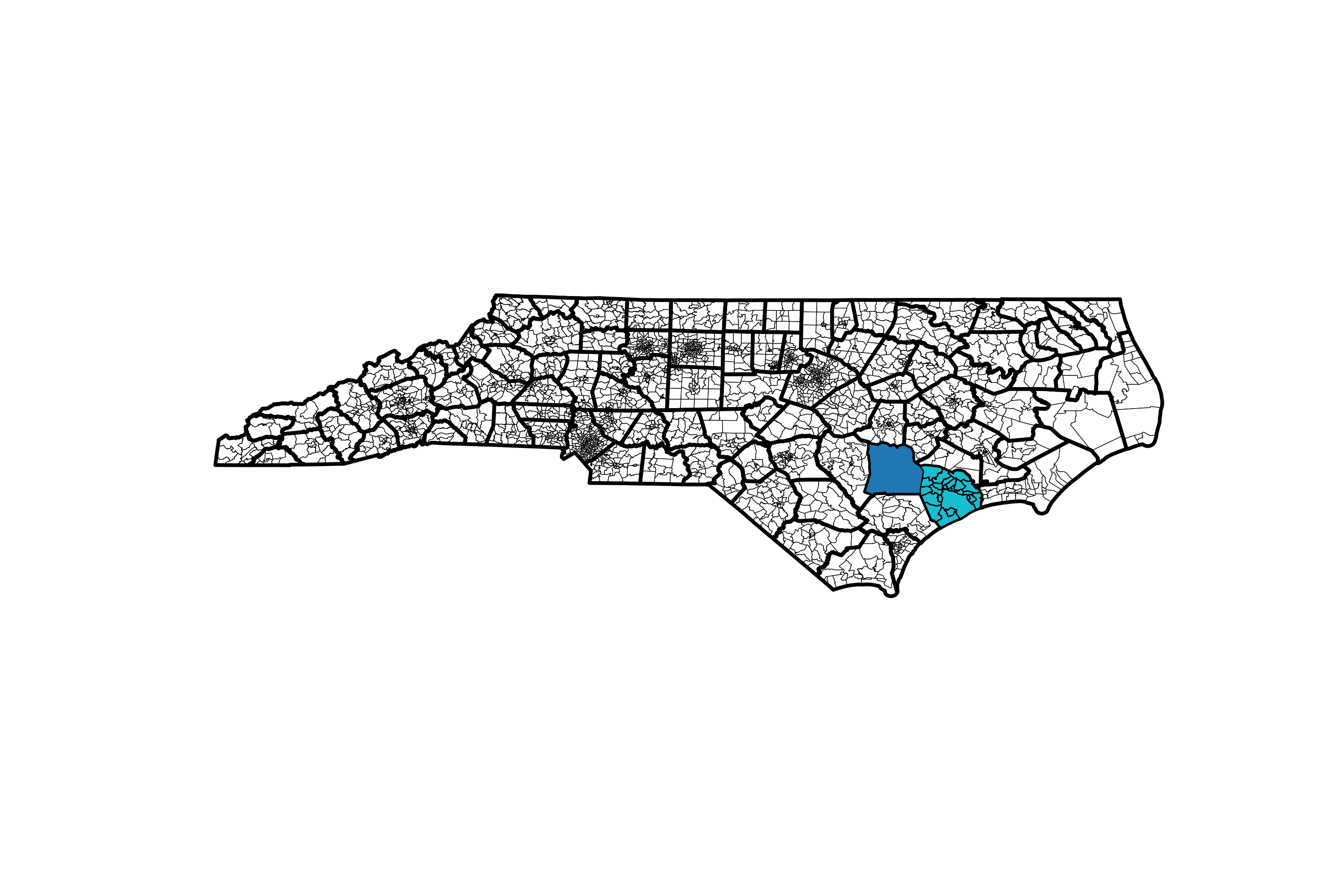}
\caption{We show a precinct map (thin black lines) with the counties (thick black lines of North Carolina. We contextualize the Duplin-Onslow county cluster with shades of blue. Duplin (dark blue; the more western county) is shown as a full county because it will not be split within the cluster; Onslow (light blue; the more eastern county) is shown with its precincts.}
\label{fig:DOinNC}
\end{figure}
The Duplin-Onslow county-cluster graph is made up of 25 nodes and 60 edges. The nodes represent the precincts used when redrawing the 2016 maps. Legally, all state House districts in North Carolina must have populations no more than a 5\% deviation from the target ideal population of a state House district. The target population is found by taking the total population of North Carolina in 2010 and dividing it by the number of statewide districts -- 120. Because of its small size, it is possible to enumerate all possible plans within this cluster that are within 5\% of the target population. We will use the 17,653 enumerated plans as an analytic bench mark in sampling with our Metropolized Forest ReCom algorithm.\footnote{The enumerated plans come through private communication with Colin Rundle; this work builds on \cite{Kawahara2017}}

As a proof of concept, we run two batches of ten independent chains. The two batches are for the two methods of choosing adjacent district pairs to merge and split (see \eqref{eq:adjdist} and \eqref{eq:bdrydist} in Section~\ref{sec:implementation}). Within each batch, each of the chains begins with a unique initial condition and runs for one million steps with $\gamma = 1$ and $J$ only accounting for the population constraints (which is constrained via the proposal).\footnote{The code we used to run these chains is available at https://git.math.duke.edu/gitlab/gjh/mergesplitcodebase.git}

The unique initial state of each chain is found in a recursive process. First, a single (uniformly random) spanning tree is drawn on the entire graph.  We then look for edges to remove such that one side of the split would be within acceptable population bounds of a single district, and the other would be within acceptable population bounds for the remaining $n-1$ districts.  The first part of the graph is designated as a district; we then draw a new spanning tree on the remaining part of the graph and repeat the above procedure.  If we do not find an acceptable edge to remove, we may either redraw the spanning tree on the remaining graph and try again (up to some maximal number of attempts) or simply start over.

With the samples collected from the chains, each starting from one of the unique initial conditions, we consider two sets of observables on this space, which are related to evaluating partisan gerrymandering. The observables are derived by selecting a set of historic vote counts and simulating state house elections by assuming the votes are fixed under changing districts. The first observable we consider is the number of elected Democrats and Republicans (between 0 and 3) within each map under a given (fixed) set of votes in each precinct. The second set of observables is the relative fraction of Democratic and Republican votes within each district; this second statistic measures not just the party representative that would win but encodes the margin of victory. We may use any historic votes to measure the above observables; in practice, we typically use many sets of votes in order to probe how the results change under different voting patterns. In the current work, we wish to test that the distribution of the MCMC chain converges to the distribution under the enumerated results. We therefore choose a historic election that would lead to a variety of elected officials from each party and find that the 2008 Governor's race is suitable for this purpose.

For the purpose of validating our algorithm we note that we are less concerned about \emph{which} low dimensional statistic we evaluate than we are with taking some relevant observable of interest.  Alternatively, we could have examined the distribution of some measure on compactness across the three districts or have taken some other relevant measure on the space.

In our experiments, we find that the methods of selecting adjacent districts lead to similar results.  Under the ``uniform district-uniform neighbor'' method (see \eqref{eq:adjdist}), an average of $185{,}884.1$ proposals per chain are accepted over the ten chains of $10^6$ steps. Under the ``boundary weighted'' method (see \eqref{eq:bdrydist}), an average of $52{,}537.2$ proposals per chain are accepted over the ten chains of $10^6$ steps.
For each batch of ten chains we plot the sampled distribution of Democrats elected in Figure~\ref{fig:HistDupOnsGv08} (left). We also estimate the rate of convergence by examining the average total variation between each of the ten chains and the enumerated ensemble. We take a best fit power law over both methods and find an order of convergence of  to be $0.32$ for the ``uniform district-uniform neighbor'' method and $0.37$ for the ``boundary weighted'' method (see right side of Figure~\ref{fig:HistDupOnsGv08}). 
At the final step, we find that the total variation between the ensemble of aggregated chains and the enumerated results has reduced to 0.0155 for the uniform district-uniform neighbor method and 0.0154 for the boundary weighted method. In both cases, the exact and sampled distributions are qualitatively very close after one million proposals.

\begin{figure}
\centering
\includegraphics[width=0.4 \linewidth]{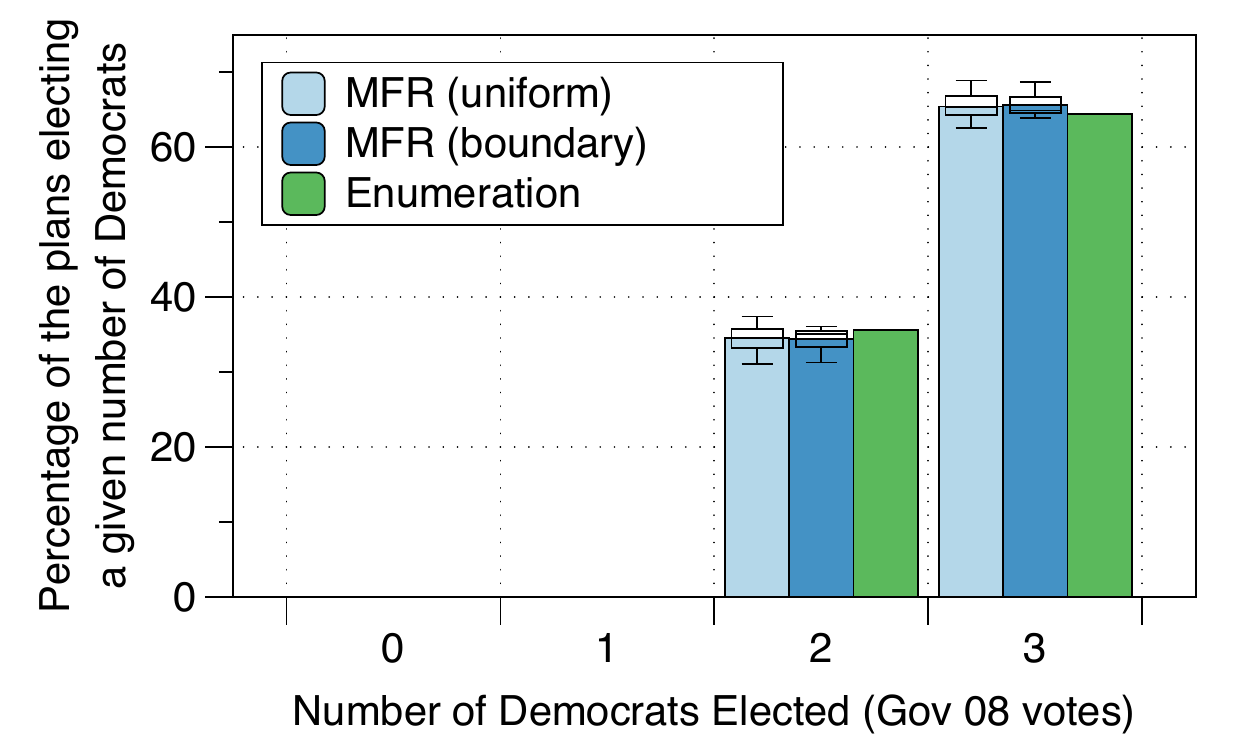}
\includegraphics[width=0.4 \linewidth]{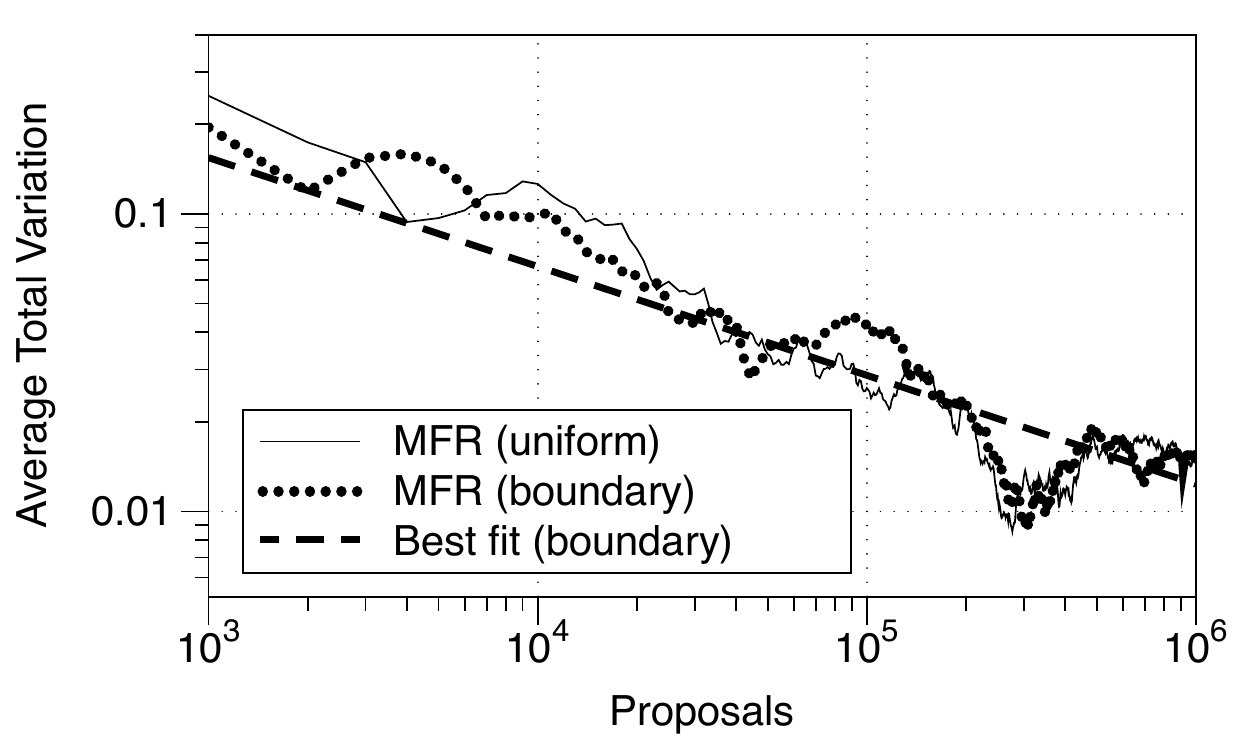}
\caption{We display the distribution of the number of Democrats elected in the three Duplin-Onslow state house districts in the case that we change the districts but fix the votes to be those of the 2008 Governor's race (left). We use two methods of picking adjacent districts when sampling. We find that the two Metropolized Forest ReCom (MFR) sampling methods (uniform district-uniform neighbor, and boundary weighted) come very close to the exact enumerated distribution (shown in green) after running ten chains of 1 million proposals. We use box plots to demonstrate the variation across chains. We then plot the averaged total variation in each of the ten distributions as a function of the number of proposals (right). We fit a power law to the observations and find that the error decays like $\propto x^{-0.32}$ and $\propto x^{-0.37}$ in the uniform district-uniform neighbor and boundary weighted methods, respectively.}
\label{fig:HistDupOnsGv08}
\end{figure}

Knowing which partisan candidate would win each district does not account for the margins of victory. To consider the margins, we examine a set of three inter-related observables: In each districting plan we order the Democratic vote fractions of the three districts from smallest to largest; we then consider the ordered marginal distributions across the ensemble of plans of the least Democratic district, the second least Democratic district, and the most Democratic district. We compare the three marginal distributions across the enumerated plans and the sampled distributions. We plot the results in Figure~\ref{fig:BoxDupOnsGv08}, again using the votes from the 2008 Governor's race. We find that the marginal distributions of the enumerated districts are extremely close to those of both ensembles.

To study the convergence of the ordered marginal distributions, we examine the rate of convergence as a function of the number of proposals.  For all pairs of chains within a batch of ten chains, we examine the total variation distance between each of the three marginal ensembles and then average this distance.  We then take the maximum averaged total variation across all pairs of chains. To measure the total variation distance, we establish histograms in each ordered marginal with a bin width of 0.2\%. We plot the averaged total variation as a function of the number of proposals in Figure~\ref{fig:BoxDupOnsGv08}. The average total variation decreases (roughly) according to a power law with order $0.39$ for the uniform-district uniform-neighbor cases and with an order of $0.34$ for the boundary weighted case.  However, the two plots are indistinguishable by eye and we only display one. After 1 million proposals the averaged total variation has decreased to less than 0.044\% for both methods and the overall features of all three distributions are extremely similar. 

\begin{figure}
\centering
\includegraphics[width=0.4 \linewidth]{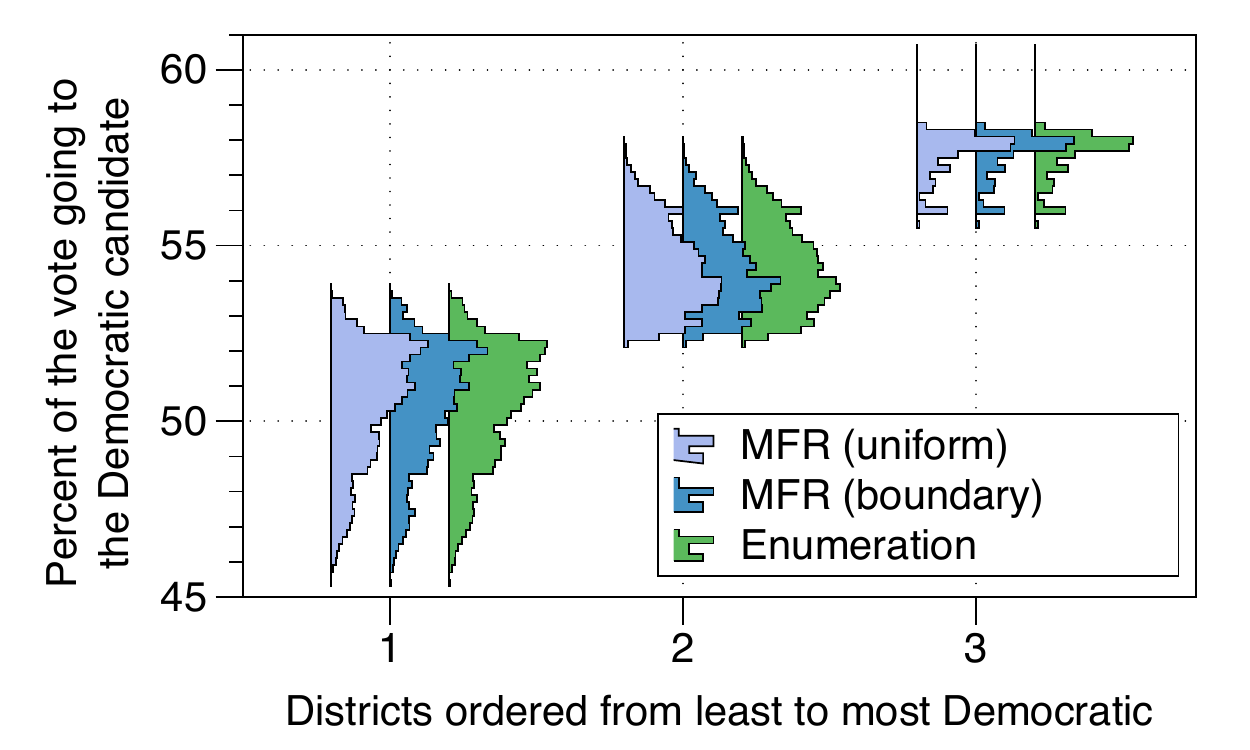}
\includegraphics[width=0.4 \linewidth]{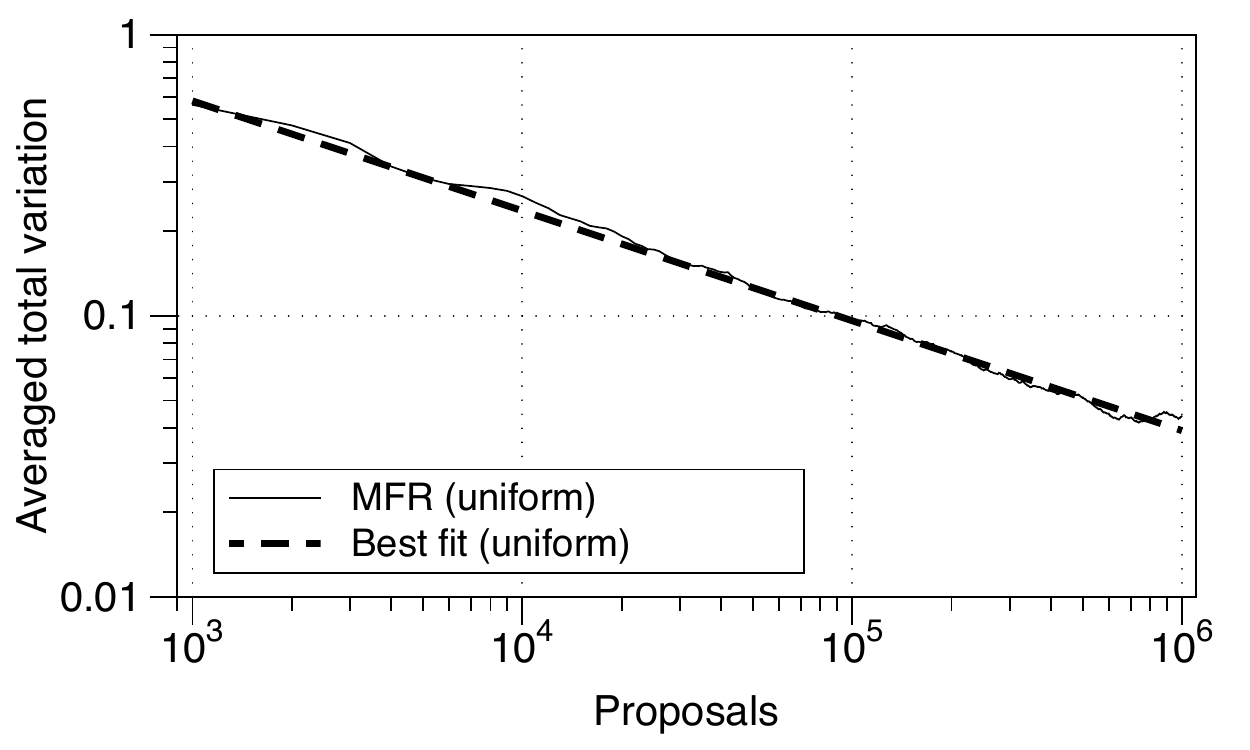}
\caption{We display the ordered marginal distributions for the percent of the vote received by the Democrat in the enumerated ensemble (green) and over the batches of 10 chains of one million proposals of the two Metropolized Forest ReCom (MFR) algorithms (dark and light blue) (left). After one million proposals, the two distributions become extremely similar. We examine the difference between the marginal distributions as a function of the number of proposals (right).  We omit the weighted boundary method as the line nearly perfectly overlaps with the uniform-district uniform-neighbor method displayed in the plot and labled ``uniform.''}
\label{fig:BoxDupOnsGv08}
\end{figure}

If we wished to further assess the partisan structure of a given plan, we may wish to examine the distribution of elected partisan officials and of marginal distributions over a collection of historic vote counts, as each vote count reveals how variations in the voting pattern effect the distributions. In principle we could have repeated the above analysis on any collection of historic votes. We may also generate the marginal distributions over demographic data, such as race, in assessing whether or not a given district's racial properties are typical of the ensemble. For example, one may consider the fraction of the black voting age population within each district. We omit such studies in the present work as we are primarily interested in demonstrating convergence of the chains rather than analyzing the partisan or demographic characteristics of the Duplin-Onslow county cluster.

Finally, we examine the true underlying distribution on the full redistricting space: namely the uniform distribution on the $17{,}653$ enumerated plans. We plot the total variation distance between the chains and the exact result as a function of the number of proposals in Figure~\ref{fig:tvAll}. 
We find a power law relationship with order $0.26$ along with a smaller constant of proportionality -- this metric converges significantly slower than the previous two estimates. By the end of the million proposals, the total variation across the ten independent chains is large after one millions proposals, with an average value of 0.47. This may reflect, the recently proven result \cite{njatDedfordSolomon2019graphs}, which states that sampling the uniform measure on graph partitions is likely NP-hard.

\begin{figure}
\centering
\includegraphics[width=0.4 \linewidth]{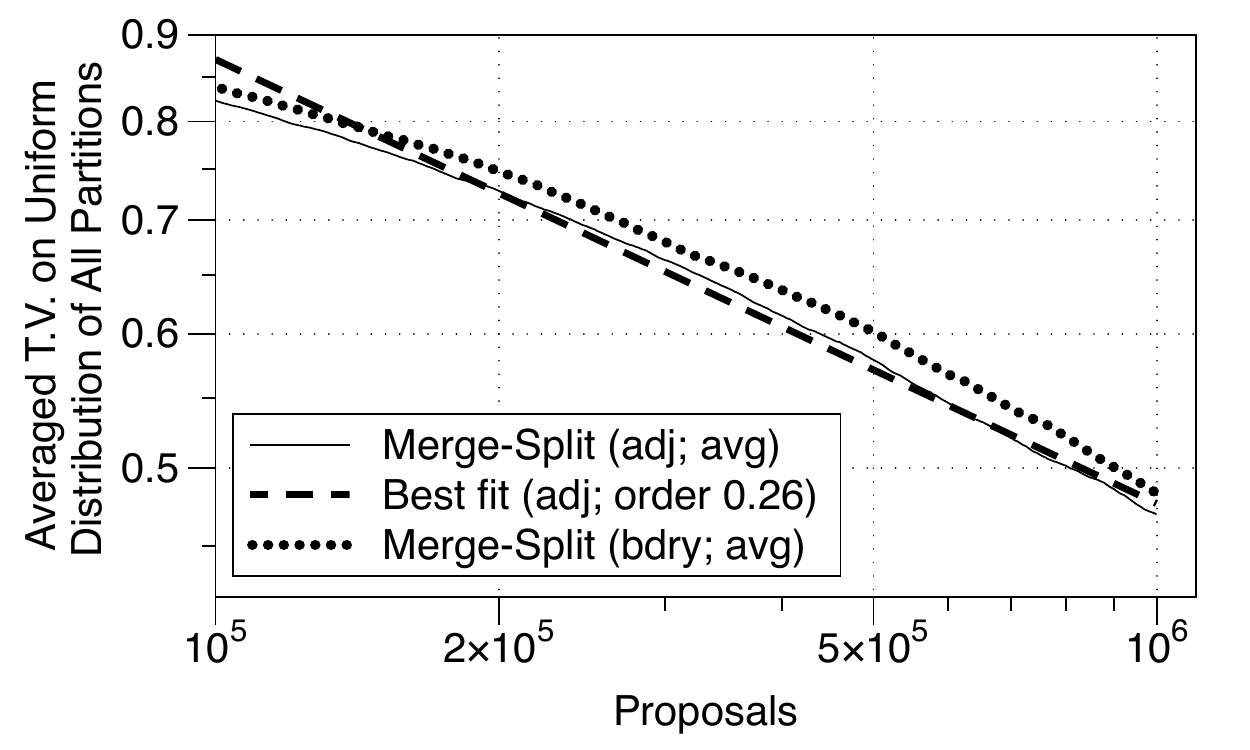}
\caption{We display the total variation between the Metropolized Forest ReCom chain and the uniform distribution on all partitions of Duplin-Onslow. On average, the ten chains hold a total variation of roughly 0.47 after one million proposals. This distribution decays according to a power law with order $0.26$ for the uniform district-uniform neighbor method and $0.25$ for the boundary weighted method.}
\label{fig:tvAll}
\end{figure}

If we place any faith in the power law relations above, then we may predict that the chain should have sampled the uniform distribution with a total variation of roughly 1\% after $2.6\times10^{12}$ proposals. However, the above experiment also demonstrates that the chains may accurately predict observables of interest far faster than they are able to recover the full 
underlying distribution. To make this comparison explicit, using the best fit power law for the distribution of elected Democrats predicts that roughly 2 million proposals would be needed to achieve a total variation of roughly 1\% (which is in the realm of what we observe). Similarly, we would need roughly 32 million proposals for the total variation in the averaged marginal distributions to be roughly 1\%. These predictions suggest that one needs a factor of roughly $10^{6}$ or $10^{5}$ more plans to approximate the distribution on all plans than the partisan outcomes considered.

We also investigate which of the enumerated plans went unsampled. Of the 10 chains in the weighted boundary method, each chain samples roughly two thirds of the $17{,}653$ enumerated plans at least once. When looking over all chains, only 795 of the enumerated plans were not sampled by any of the 10 chains after the 1 million proposals. We examine the histogram of the product of spanning trees, $\tau(\xi)$, for the enumerated plans and the sampled plans in Figure~\ref{fig:stcounthistsDO}.  We find the distributions are nearly identical, however there is some variation in for plans with fewer total spanning forests.  To confirm this observation, we plot the histogram of enumerated plans that were sampled at least once and compare it to the histogram of enumerated plans that were not sampled and see that, in general, the unsampled plans have fewer associated spanning forests (see Figure~\ref{fig:stcounthistsDO}).

\begin{figure}
\centering
\includegraphics[width=0.4 \linewidth]{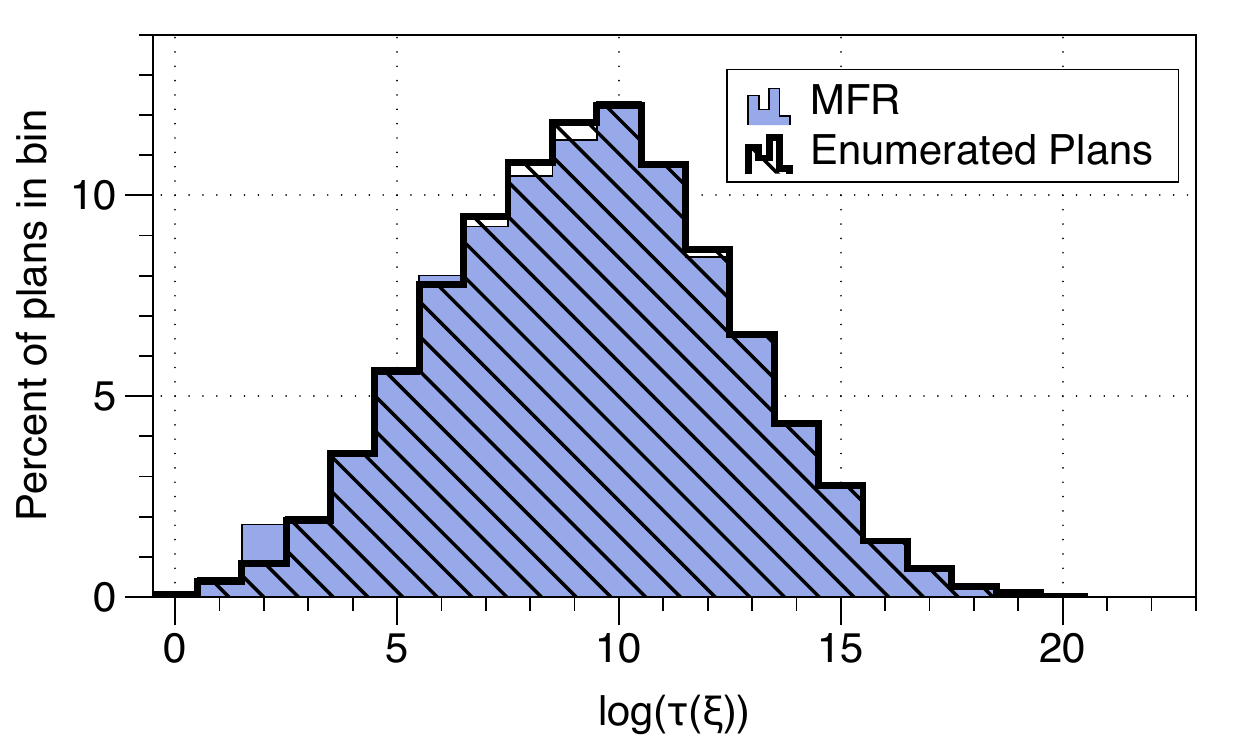}
\includegraphics[width=0.4 \linewidth]{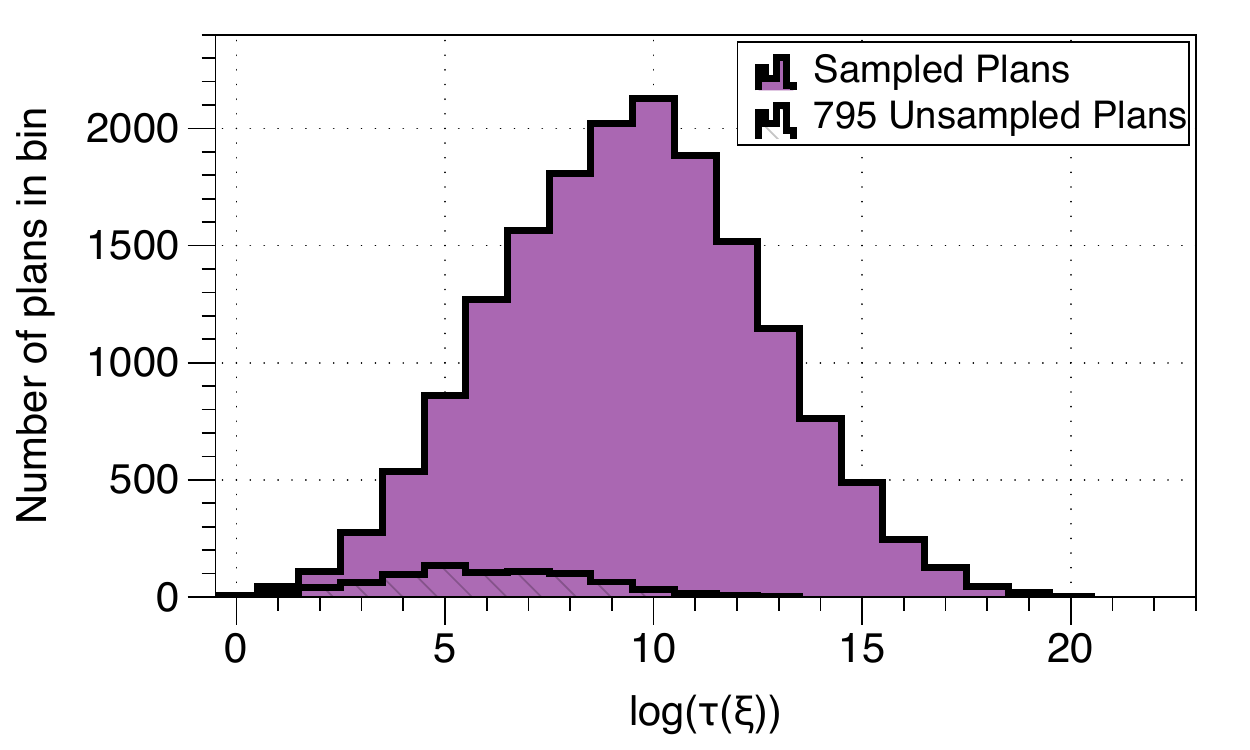}
\caption{On the left, we plot the histogram of the number of spanning forests, $\tau(\xi)$, for the enumerated plans and the sampled plans.  On the right, we examine the histogram of the plans that were sampled at least once and compare it with the plans that we not sampled within the boundary weighted Metropolized Forest ReCom (MFR) batch of 10 chains.
}
\label{fig:stcounthistsDO}
\end{figure}

\subsubsection{Acceptance rate dependence on \texorpdfstring{$\gamma$}{gamma}}
To conclude the Duplin-Onslow example, we examine the acceptance rates as a function of $\gamma$. 
We generate eight additional chains on the uniform district-uniform neighbor method, the $i$th of which has a $\gamma = (i-1)/8$. 
The score function on these chains only considers the population deviation, which is handled via the proposal; for Duplin-Onslow we again use a threshold of 5\%.
We run the chains for one million proposals and estimate the acceptance rates in Figure~\ref{fig:acceptance}. As discussed previously, we see the highest acceptance rates for low $\gamma$ and the lowest acceptance rates for large $\gamma$. The relationship between acceptance with $\gamma$ appears to be monotonically decreasing. For the Duplin-Onslow graph, it yields an acceptance of just over 40\% for $\gamma = 0$, and just under 20\% for $\gamma = 1$.

\begin{figure}
\centering
\includegraphics[width=0.4 \linewidth]{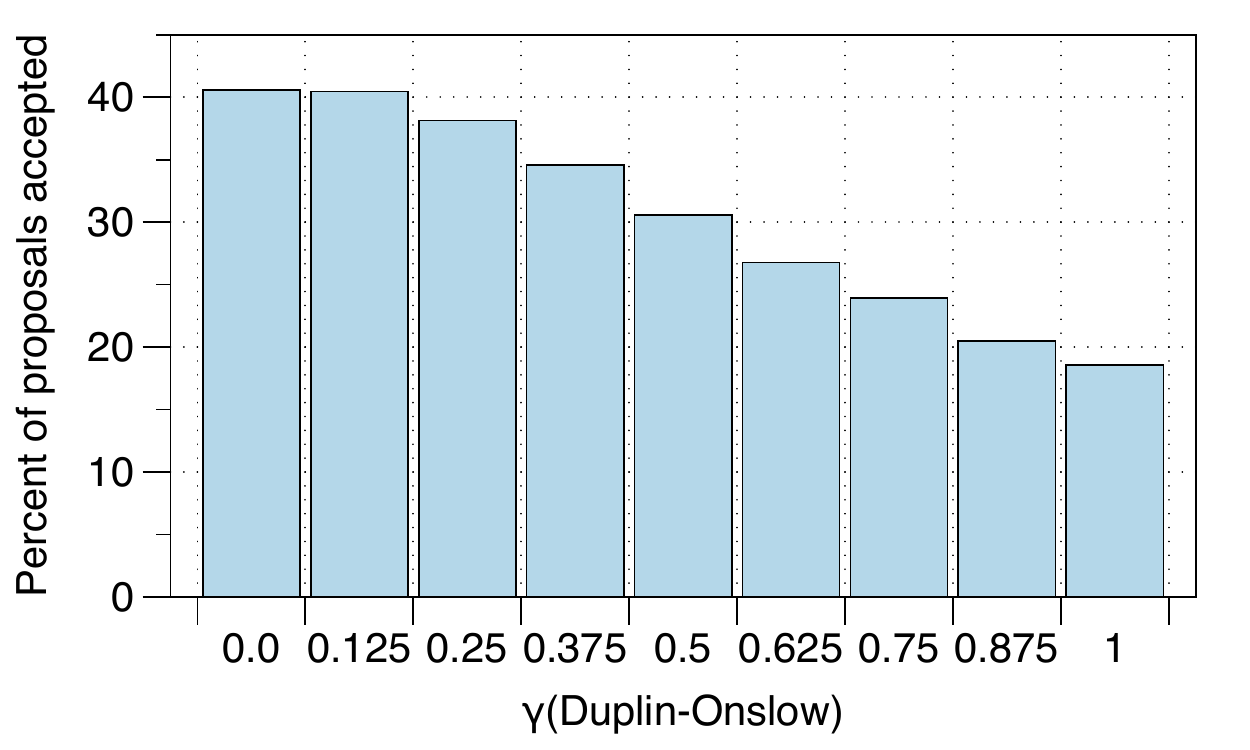}
\caption{We display the acceptance rate in the Metropolized Forest ReCom algorithm as a function of $\gamma$ in the Duplin-Onslow county cluster under the uniform district-uniform neighbor pair proposal.}
\label{fig:acceptance}
\end{figure}

\subsection{North Carolina Congressional Districts} We continue by examining the Metropolized Forest ReCom algorithm applied to a larger problem.  We examine the 13 congressional districts of North Carolina.  Our energy only encompasses a hard population constraint of 1.5\% deviation. Legally compliant population deviations for congressional districts are significantly smaller than this, however, we have previously shown that zeroing out populations with similar deviations leads to negligible shifts in the two observables of interest (see \cite{herschlag2020quantifying}).  We remark that the plans we will sample will not account for other redistricting criteria such as fulfilling the voting rights act and preserving counties; as a consequence, care should be taken in using the ensembles below in discussions around policy. Instead, we intend this to be an investigation of the algorithm

Without constraints, districts sampled from the space of uniform (population compliant) partitions may have jagged boundaries and form highly non-compact districts that look like space filling curves (see, e.g. \cite{deford2019recombination}).  Under the ReCom Markov chain, districts tend to be significantly more compact, however, built in rejections for non-compact plans may be included \cite{deford2019recombination}. As noted above, when $\gamma = 0$, our Metropolized Forest ReCom is a Metropolized version of ReCom with an expanded state space; because of this ReCom should focus sampling on districts that are comparable to the forest distribution with $\gamma = 0$. We therefore begin by examining convergence properties of Metropolized Forest ReCom for $\gamma = 0$ with no additional constraints (other than population) and with a constraint that ensures no district within a plan has an isoperemetric ratio of more than 110 (or Polsby-Popper less than 0.114; this threshold allows districts to be slightly less compact than the least compact district in the 2016 North Carolina congressional districts).

We examine the convergence properties of both cases by running a batch of ten chains for each method, each starting with randomized initial conditions, and using the ``boundary-weighted'' method of choosing adjacent districts. When bounding the district compactness, we ensure that the initial states contain districts that are all more compact than the allowable limit. When unconstrained, we find an acceptance rate of just over a fourth: Each chain accepts more than $260{,}000$ proposals and accepts $265{,}146.3$ proposals on average. When ensuring compact districts, we find an acceptance rate of just over a fifth: Each chain accepts more than $210{,}000$ proposals and accepts $211{,}966.3$ proposals on average.

To test the convergence, we use votes cast in the 2012 presidential race and plot the histogram of the number of Democrats that would have been elected in the ensemble of districting plans under these votes (see Figure~\ref{fig:HistNC}). We find that the distribution of elected Democrats converges to nearly the same distribution, independent of the starting location.  We also examine the rate of convergence by examining the maximum pairwise total variation across all histograms of elected Democrats after a given number of proposals.  Taking a best fit power law, we find that the order of convergence is $O(0.52)$ without any constraint on compactness and $O(0.4)$ with the constraint (see the right of Figure~\ref{fig:HistNC}).

\begin{figure}
\centering
\includegraphics[width=0.4 \linewidth]{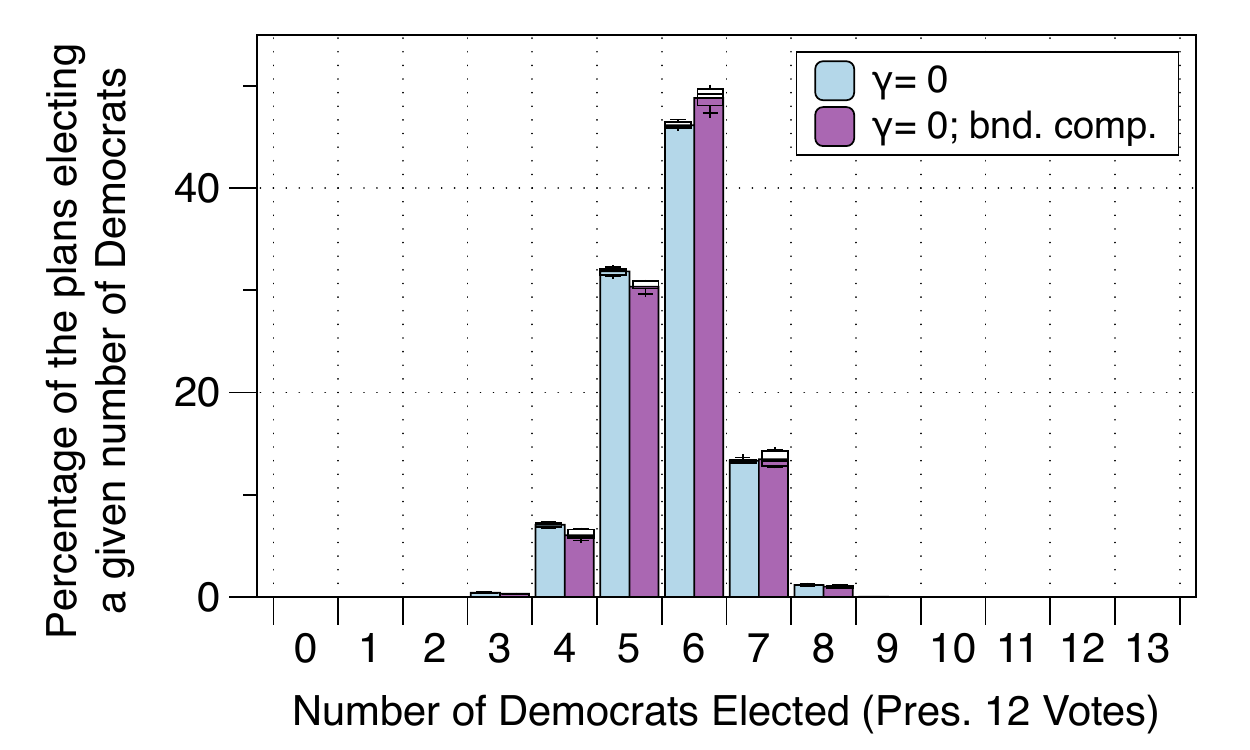}
\includegraphics[width=0.4 \linewidth]{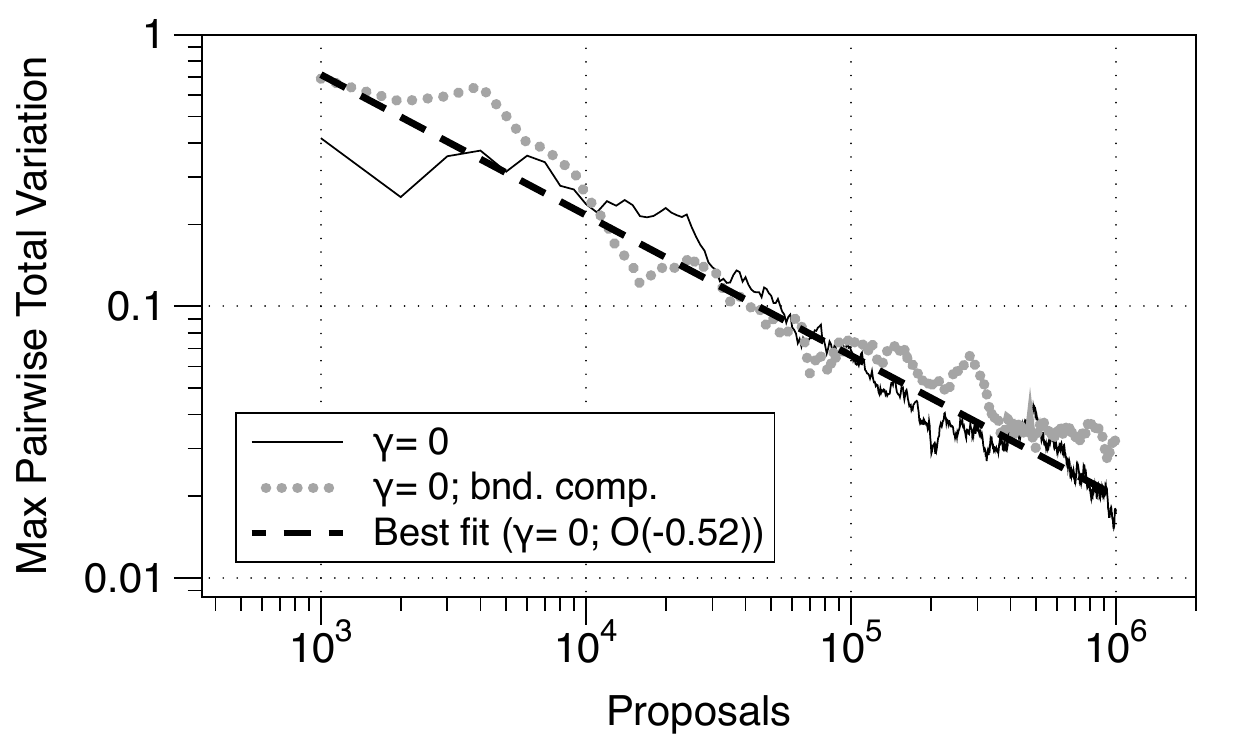}
\caption{We display the distribution of the number of Democrats elected in the 13 districts in the North-Carolina congressional districts in the case that we change the districts but fix the votes to be those of the 2012 Presidential race (left). We use box plots to demonstrate the variation across chains becomes small after 1 million proposals. We plot the maximal pairwise total variation between the ten distributions as a function of the number of proposals (right). We fit a power law to the observations and find that the error decays like $\propto x^{-0.52}$ when $\gamma = 0$ and $O(0.4)$ when $\gamma = 0$ and with constrained compactness.}
\label{fig:HistNC}
\end{figure}

We continue by examining the ordered marginal distributions of Democratic vote fractions across the ensemble of plans.  We create histograms for the marginal distributions with a bin width of 0.2\%. We then compute the average total variation distance across all 13 marginal ensembles between each pair of chains and find the pair with the maximum average total variation distance across all pairs. We display the two most distinct chains after 1 million proposals along with the distribution comprised of all 10 chains in Figure~\ref{fig:BoxNC12}, for the batch that constrains compactness; we do this because the chain that does not constrain compactness has even smaller differences across chains. The 10 chains converge to nearly identical distribution, both when constraining the compactness or not.  We take a best fit power law by examining the maximum pairwise distribution as a function of the number of proposals and find a convergence rate of $0.5$ when the compactness is unconstrained and $0.44$ when the compactness is bounded.

\begin{figure}
\centering
\includegraphics[height=3.5cm]{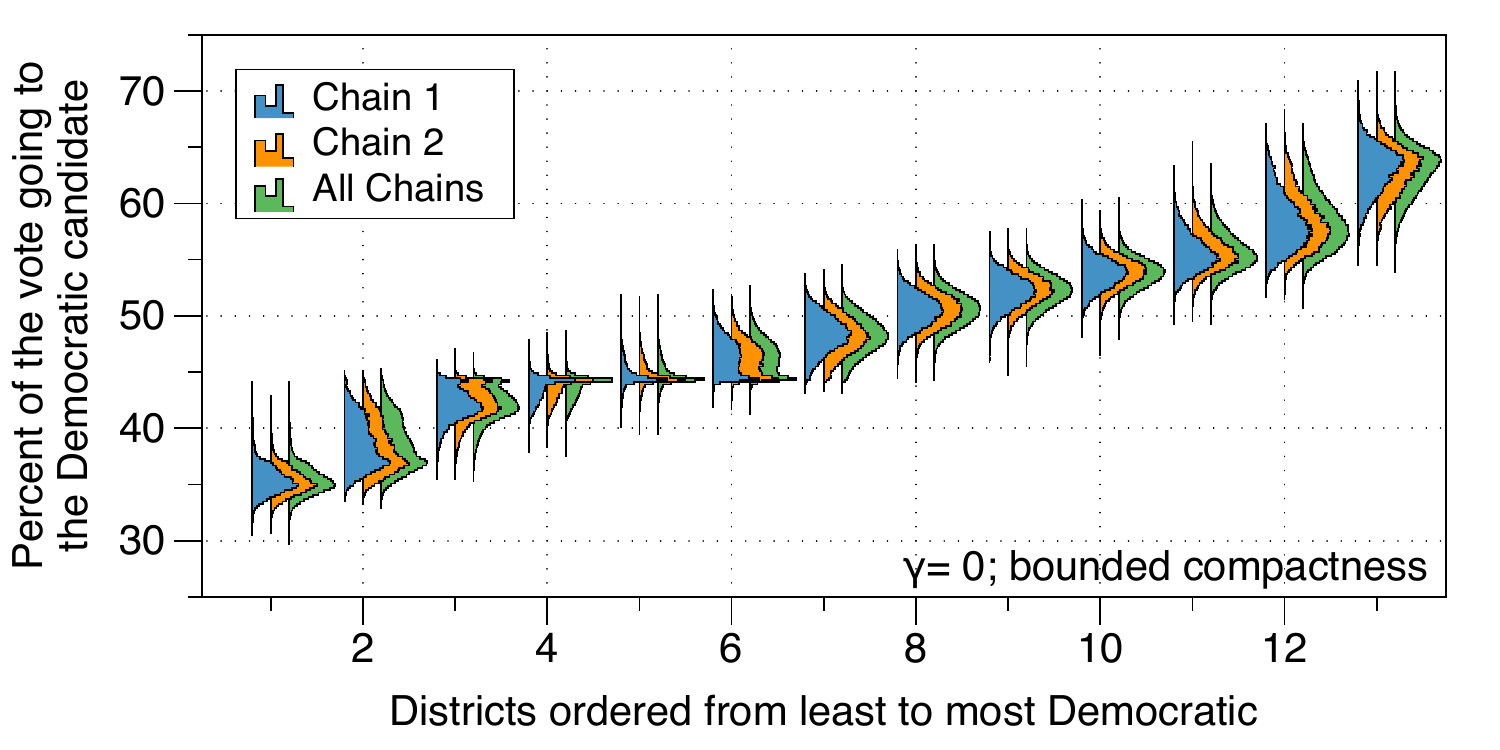}
\includegraphics[height=3.5cm]{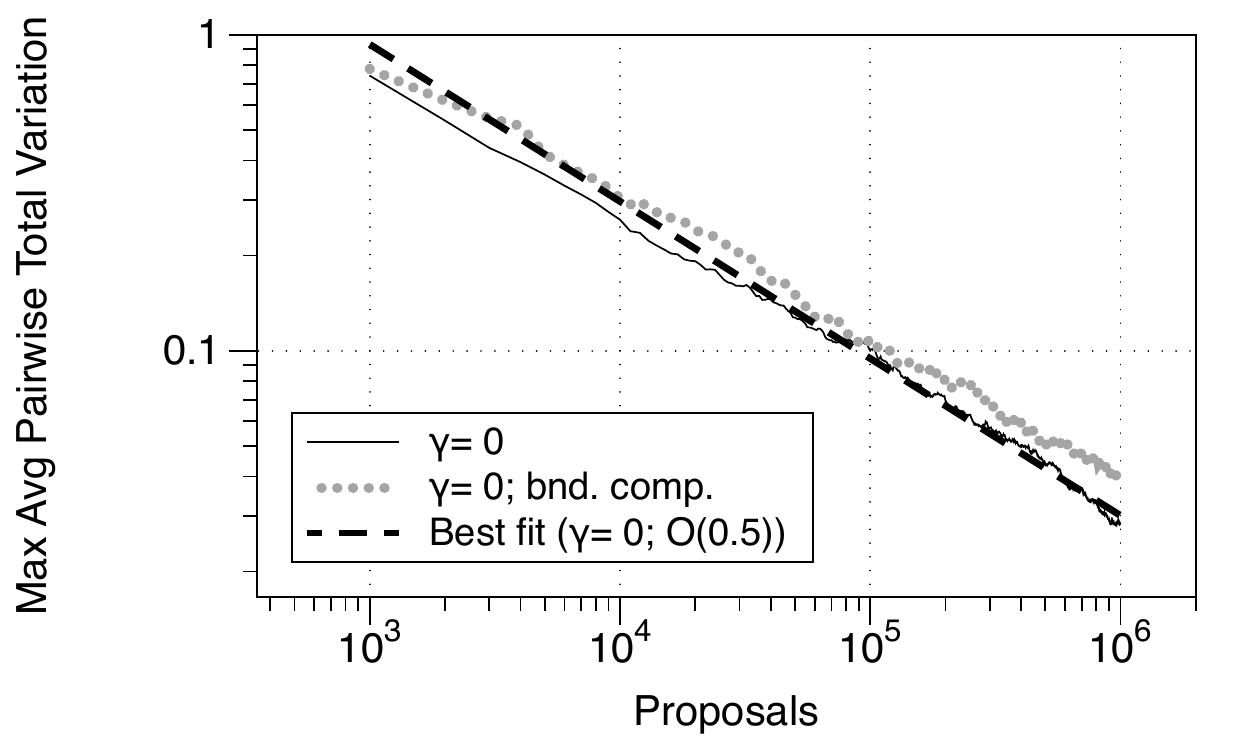}
\caption{We display the ordered marginal distributions for the percent of the vote received by the Democrat in all ten chains (green) and over the two chains that are most distinct after one million proposals (orange and blue) for the batch with $\gamma = 0$ and constrained compactness (left). After one million proposals, the distributions become extremely similar. We then take the maximum pairwise total variation averaged across all marginal distributions as a function of the number of proposals for the batches with and without compactness constrains with $\gamma = 0$ (right). We find a rate of convergence on the order of $0.5$ when there is no constraint on compactness and $0.44$ when there is a constraint.}
\label{fig:BoxNC12}
\end{figure}

\subsubsection{Changing the invariant measure}
In setting $\gamma = 0$, we concentrate on plans with higher numbers of spanning forests $\tau$.  Plans with higher numbers of spanning forests tend to be more compact, however, redistricting committees may be reluctant to use the number of spanning trees of a district as a measure of compactness.  A more traditional measure of compactness is the Polsby-Popper score, which is the inverse of the isoperemetric ratio scaled to vary from zero to one; more precisely, the Polsby-Popper score is a district's area divided by the square of its perimeter multipled by $4\pi$.  We examine the correlation between the Polsby-Popper score and the spanning tree counts in the context of single districts and entire plans in Figure~\ref{fig:stvpp}.  We find a positive correlation of only 0.41 between the number of spanning trees within a district and the district's Polsby-Popper score, and a correlation of 0.2 between the total number of spanning forests defining a plan (partition) compared to the average Polsby-Popper score of the plans (partitions).

\begin{figure}
\centering
\includegraphics[width=0.4 \linewidth]{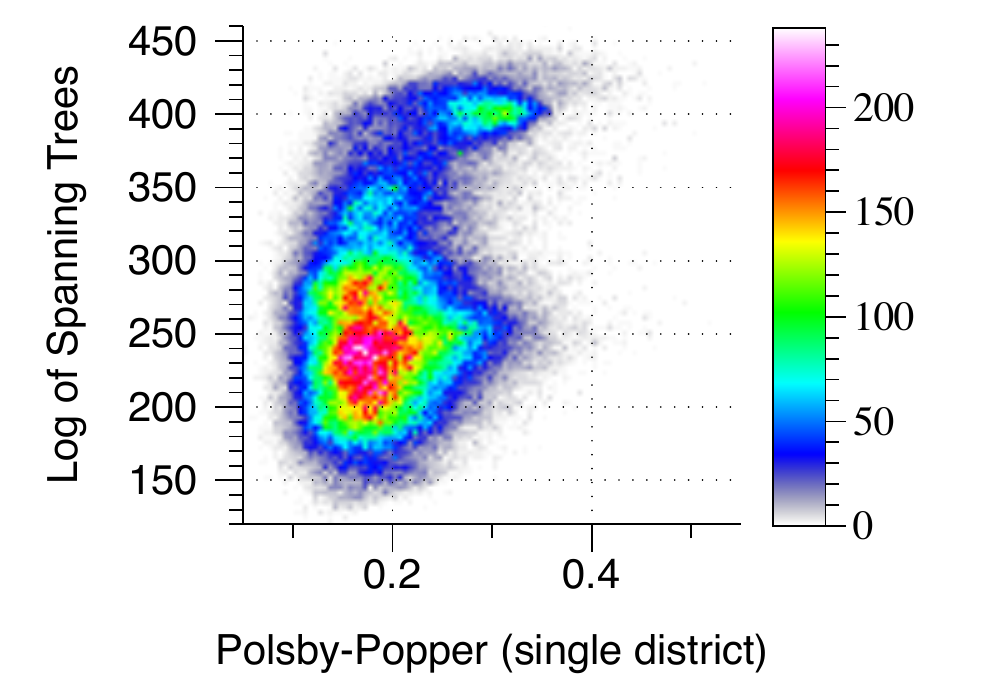}
\includegraphics[width=0.4 \linewidth]{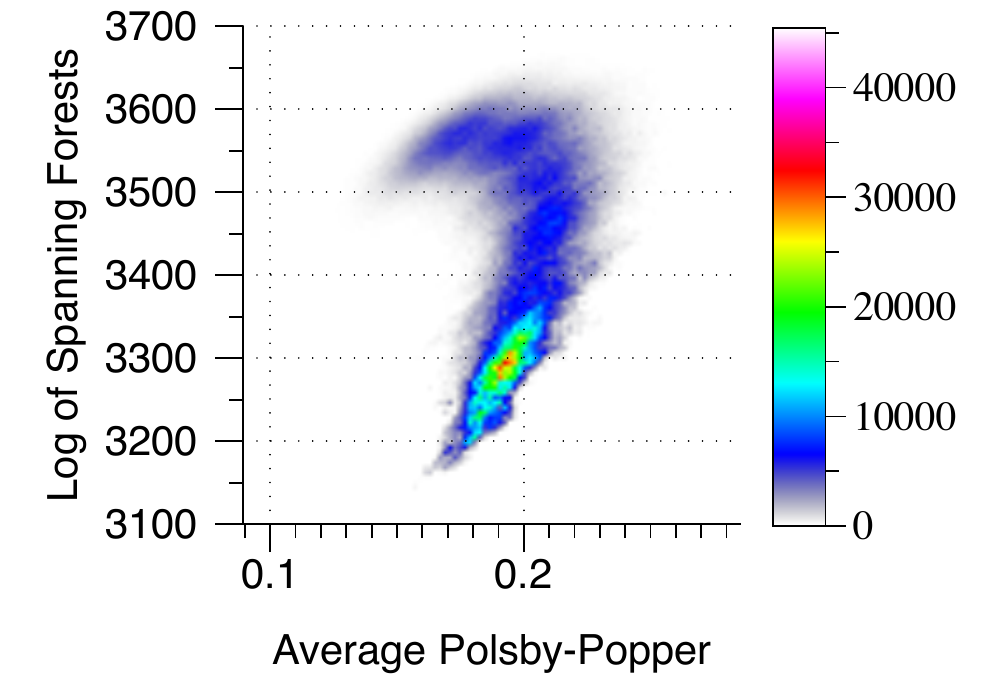}
\caption{We display heat maps for the number of spanning trees in a district with respect to the district's Polsby-Popper score (left).  We also compare the number of spanning forests in a plan, $\tau(\xi)$, with respect to the average Polsby-Popper score of the districts within the plan (right).  In the former case, we find a correlation between these two measure of 0.41, and in the latter we find a correlation of 0.2.}
\label{fig:stvpp}
\end{figure}

It is generally understood that different compactness measures may lead to different conclusions \cite{barnes2018gerrymandering}, but it is unclear how tenable spanning tree/forest counts could be as a measure of compactness. When $\gamma = 0$, the measure places more importance on plans with higher spanning forest counts.  One may wish to remove this relative importance between plans (which is to set $\gamma = 1$).
One such avenue to remove the preference of plans with higher spanning forest counts is to sample at $\gamma = 0$ and to then use importance sampling, however this will almost certainly not work due to vast differences $\tau(\xi)$ across different plans; indeed, in a chain with $\gamma = 0$ and bounded compactness, the spanning forest counts within a standard deviation differ by an order of $10^{23}$ (e.g. see below in Figure~\ref{fig:ppvprop})). One could instead attempt to sample at $\gamma = 1$ while utilizing a score function $J(\xi)$ which concentrates around compact plans.

Setting $\gamma = 1$ also allows one to efficiently mix any proposals from Forest ReCom with those from a single node flip algorithm. In this case, one could ignore the trees when measuring relative preferences/probabilities between plans or remove the trees to evolve a single node flip algorithm.  In this case, spanning trees could be re-initialized as needed, depending on whether a Forest ReCom or single node flip proposal was being considered.

Therefore we investigate the algorithm at $\gamma = 1$, and begin with no constraints or additions to the score function. In this case, we see no evidence of convergence after 10 million steps as the total variation distance between chains remains very high in both observables. The reason for this may be that uniformly drawn partitions on this graph may have undulating boundaries with highly non-compact districts (e.g. see \cite{deford2019recombination}); the consequence would be that when districts are merged along a boundary, the proposal will tend to smooth the boundary which will lead to a higher number spanning trees in the proposed districts and lead to a higher chance of rejection.  Indeed, we see that the compactness greatly decreases as the chain evolves in Figure~\ref{fig:ppvprop} which supports this hypothesis.

We continue to investigate the $\gamma = 1$ case by adjusting the score function $J(\xi)$ to include a compactness term.  We do this by adding a sum of the isoperimetric ratios of the districts and weight this sum by a factor $w_c$ so that now
\begin{align}
J(\xi) &= J_{pop}(\xi) + w_c J_c(\xi)\quad\text{where}\\
J_c(\xi) &= \sum_{d = 1}^D P_d^2 / A_d,\quad
\end{align}
and $J_{pop}(\xi)$ is infinite when the population deviation of a district is beyond an allowable threshold (i.e. leads to an unbalanced partition), $P_d$ is the perimeter, or boundary length, of district $d$, and $A_d$ is the area of district $d$.  We tune $w_c$ so that we are likely to sample districts that are as compact as those sampled with $\gamma = 0$ and find an acceptable value of $w_c = 0.45$. We show the similarity of the average compactness scores between the two $\gamma = 0, w_c = 0$ chains (one with constraints; one without) and $\gamma = 1, w_c = 0.45$ in Figure~\ref{fig:ppvprop}.  In the same figure, we also examine the evolution of the least compact district in these runs and find that the run with $\gamma = 1, w_c = 0.45$ leads to more compact least-compact districts than the run with $\gamma = 0, w_c = 0$. When $\gamma = 1$ and $w_c = 0$ we see the plans become extremely non-compact.  We also display the number of spanning trees. Changing $w_c$ from 0 to 0.45 reduces the discrepancy in the number of spanning trees when compared to the $\gamma = 0$ chains. However, the number of spanning trees found in the $\gamma = 1, w_c = 0.45$ chain are still significantly smaller than those of the $\gamma = 0$ chains. When $\gamma =1$ and $w_c = 0.45$, we see no evidence of convergence after 1 million steps.

\begin{figure}
\centering
\includegraphics[width=0.3 \linewidth]{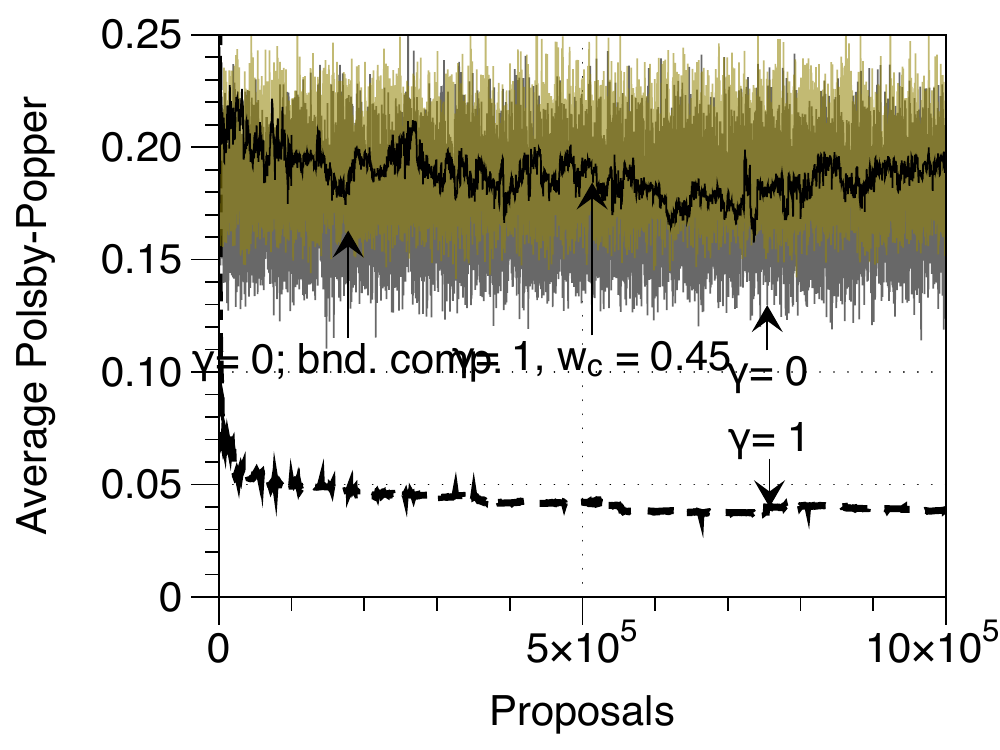}
\includegraphics[width=0.3 \linewidth]{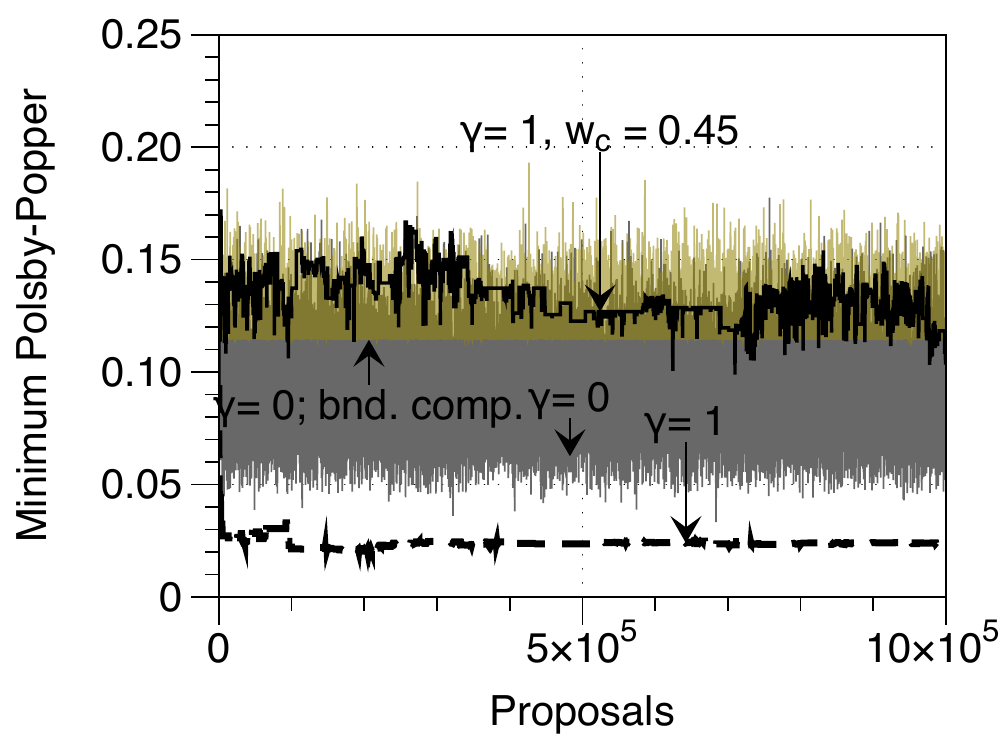}
\includegraphics[width=0.3 \linewidth]{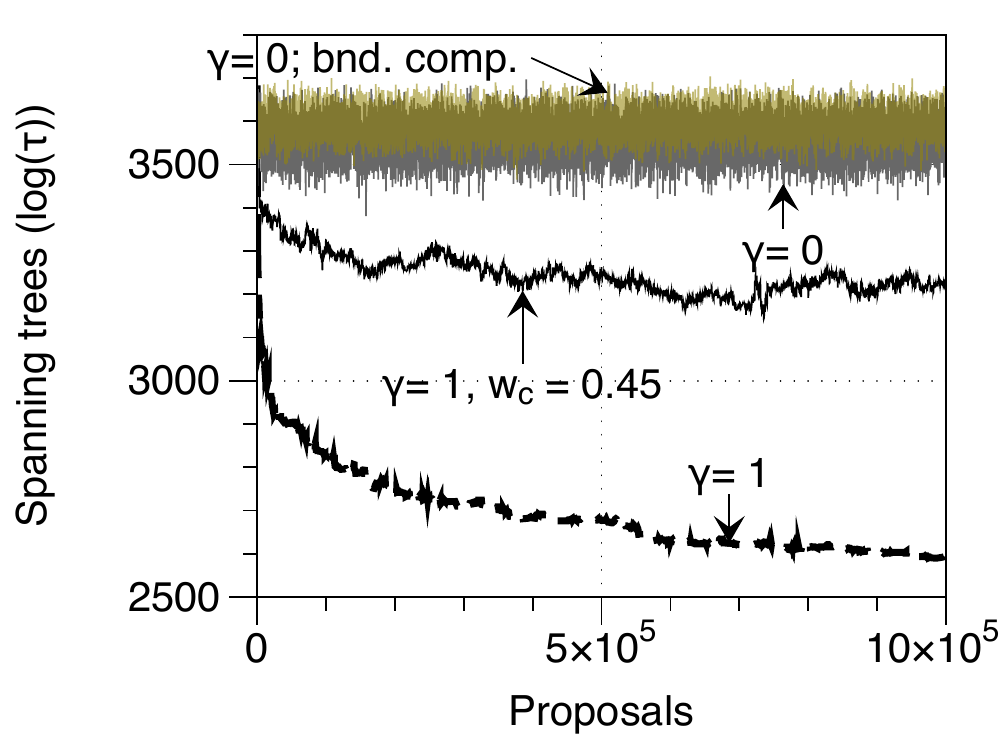}
\caption{We display the average Polsby-Popper scores over a plan, the minimal Polsby-Popper score over all districts within a plan, and the total number of spanning forests as a function of the number of proposals on the left, center, and right, respectively. The gold line represents the $\gamma = 0$ case with bounded compactness; it is semi-transparent and is displayed over the gray line that represents the $\gamma = 0$ case without bounded compactness.  In the center, the gold line is strictly bounded below and thus appears flat.}
\label{fig:ppvprop}
\end{figure}


Having observed convergence for $\gamma = 0$ and a lack of convergence at $\gamma = 1$, we conclude by examining the potential for a tempering algorithm to traverse between these two regions.  There are two primary methods of tempering: simulated and parallel.  
In parallel tempering, we consider the product measure of $N$ chains. One chain has a marginal distribution that mixes quickly, and another samples from the target distribution.  The remaining chains are used to allow chains to swap states via a swap proposal that is made reversible through Metropolis-Hastings.  The idea is that a state may start in the chain that quickly mixes, but may gradually swap chains until it finds itself in the chain with the target distribution. When it is in the chain with the target distribution, it will be used to collect observables of interest.

In the current context, we could temper by expanding the state space to include a variety of $\gamma$'s and $w_c$'s that linearly interpolate between a known fast mixing chain ($\gamma = w_c = 0$) and a target distribution ($\gamma = 1$, $w_c = 0.45$).  One must provide a sufficient number of intermediate values so that the states can swap chains and move to and from both extremes.

To investigate the potential efficacy of tempering, we examine nine batches of ten chains, each run for one million proposals, and linearly interpolate between parameters so that $\gamma = i/8$ and $w_c = 0.45\cdot i/8$ for $0\leq i \leq 8$.  We first examine the convergence properties of the nine batches chains by examining the maximum pairwise average total variation distance between all chains within a batch.  We plot the result in Figure~\ref{fig:TVvsGamma}, and find that the distributions gradually (close to linearly) diverge as we simultaneously increase $\gamma$ and $w_c$ until the deviation is close to 1 at $(\gamma,w_c) = (1, 0.45)$. To examine the possible efficacy of tempering, we examine how the distributions of spanning trees and average Polsby-Popper score overlap within a single chain in each batch.  We only consider elements in the chains after $200{,}000$ proposals and plot both sets of distributions in Figure~\ref{fig:overlapingDistsTemp}. We find that both the distributions of spanning trees and distributions of Polsby-Popper have significant overlap across successive increments as $\gamma$ and $w_c$ increase/decrease.  We also examine the probability of seeing the samples from each chain according to the target measure with $\gamma = 1$ and $w_c = 0.45$; we see that the chains at different parameters sample from much different locations in the space, but that these regions overlap as the parameters vary.

\begin{figure}
\centering
\includegraphics[width=0.4 \linewidth]{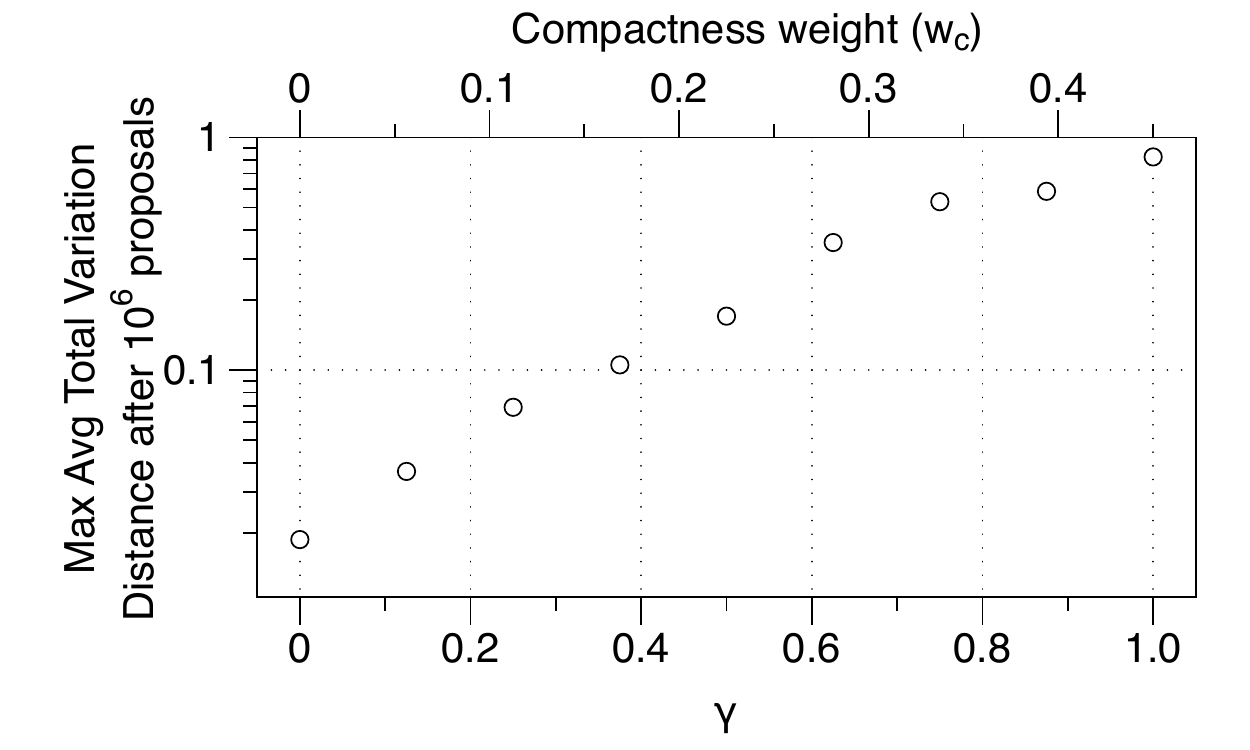}
\caption{We display the maximum pairwise total variation distance across 10 chains after 1 million proposals for different values of $\gamma$.}
\label{fig:TVvsGamma}
\end{figure}

\begin{figure}
\centering
\includegraphics[width=0.3 \linewidth]{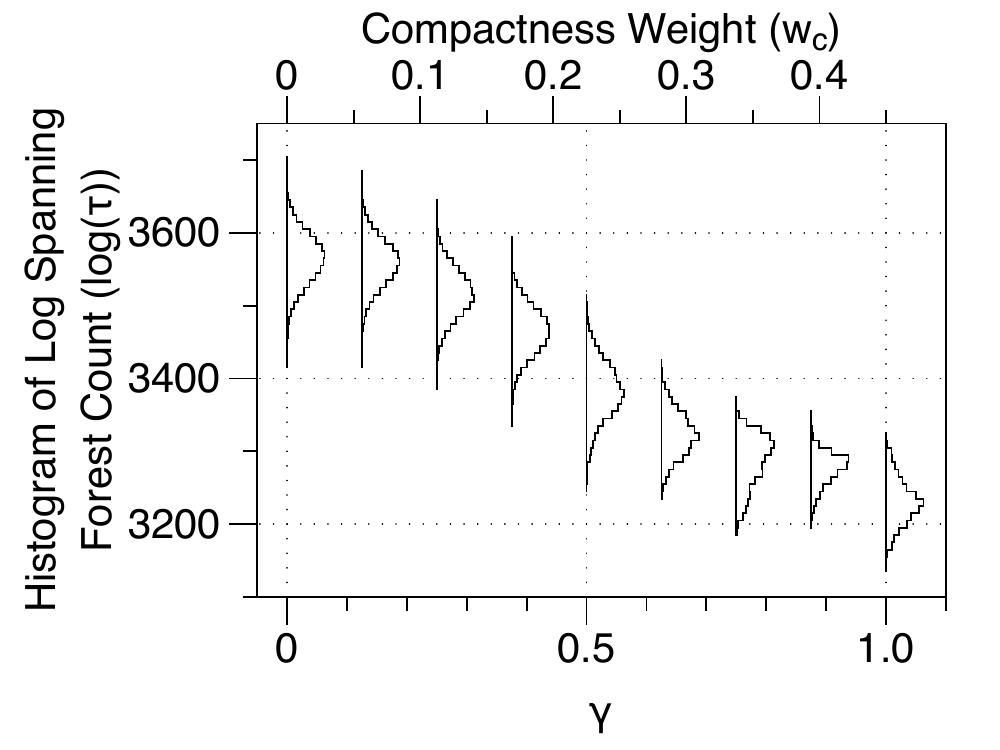}
\includegraphics[width=0.3 \linewidth]{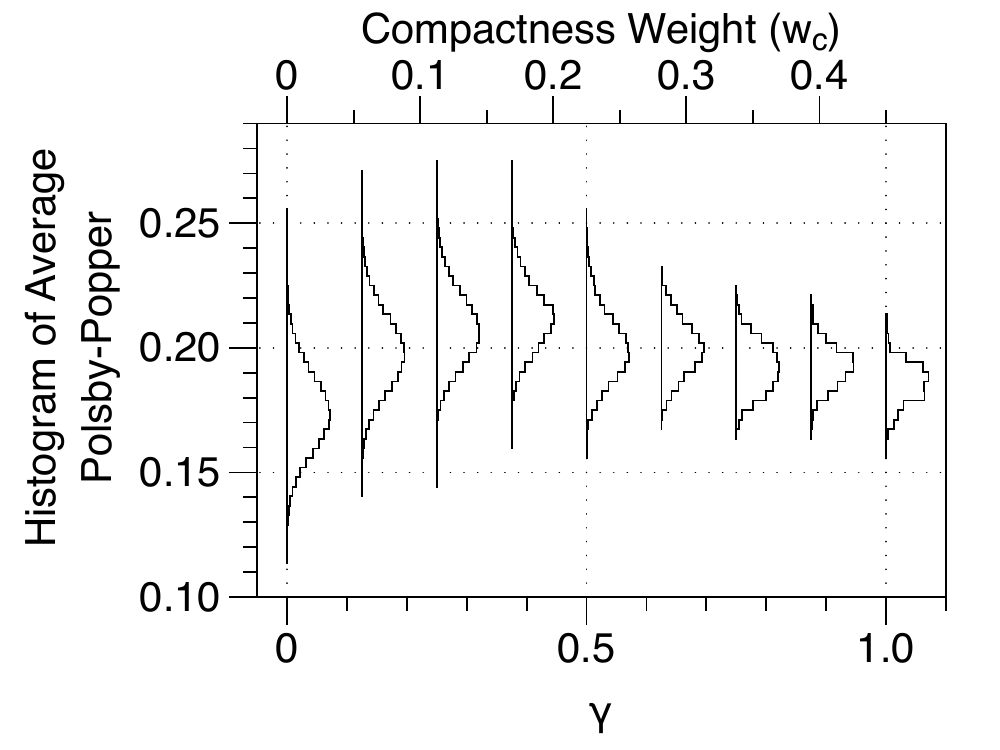}
\includegraphics[width=0.3 \linewidth]{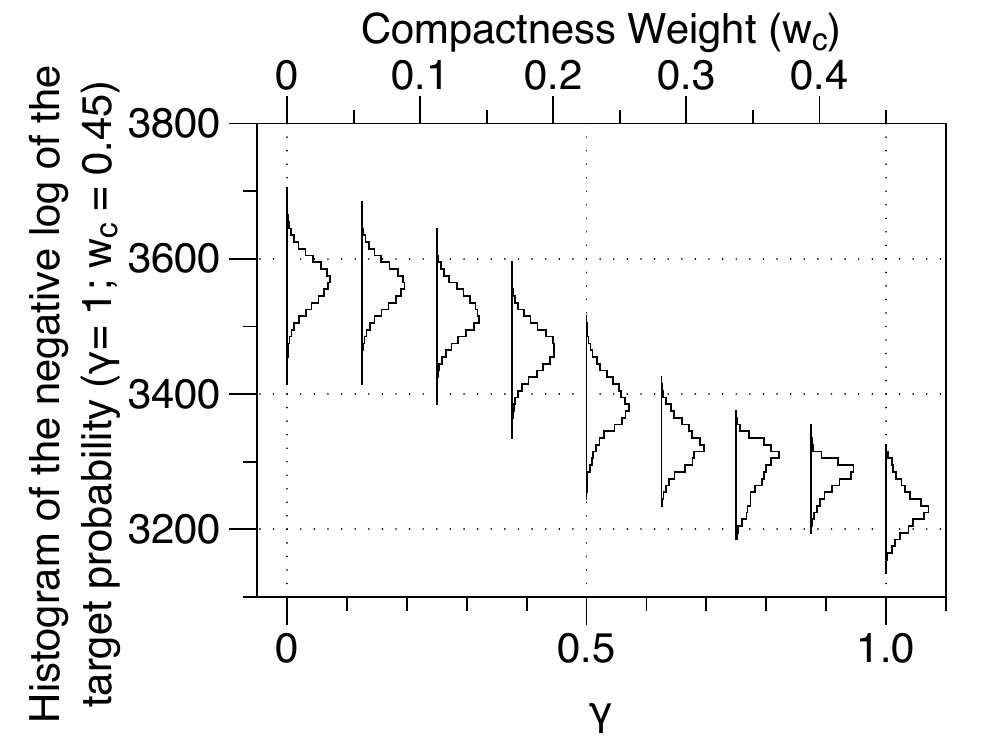}
\caption{We display the distributions of the spanning trees (left) and compactness (right) over a variety of parameters $(\gamma, w_c)$ that linearly interpolate between $(\gamma, w_c) = (0,0)$ and $(\gamma, w_c) = (1,0.45)$.  Finally, we also display the probability of the sampled plans from each chain with respect to the target measure of $(\gamma, w_c) = (1,0.45)$.}
\label{fig:overlapingDistsTemp}
\end{figure}

Each of the chains shown in Figure~\ref{fig:overlapingDistsTemp} has the same initial state.  Given the lack of evidence in convergence at higher values of $\gamma$ and $w_c$, it is possible that the chains are sampling around some local minimum.  If the chains are stuck around this local minimum, then the overlap of distributions we see may not hold for other minima; furthermore, we have only visualized marginal distributions related to the energy (in spanning forest counts and compactness scores) rather than joint distributions of these values.  

To further probe the viability of a tempering scheme, we analyze how the chains would mix in the context of parallel tempering.  Using the nine batches of 10 chains from above, with the parameters of each chain defined as $(\gamma = i/8, w_c = 0.45\cdot i/8)$, for $0\leq i\leq 8$, we introduce a tempering proposal that swaps states across chains. For example, if the the state space of the 9 chains is given as $(\xi_0, \xi_1, ..., \xi_i, \xi_{i+1}, ... \xi_8)$, we may propose a new state $(\xi_0, \xi_1, ..., \xi_{i+1}, \xi_i, ... \xi_8)$.  Under Metropolis-Hastings, the acceptance probability of swapping states across chains may is calculated as 
\begin{align*}
A((..., \xi_i, \xi_{i+1}, ...), (..., \xi_{i+1}, \xi_{i}, ...)) &= \min\Bigg(1, \frac{\pi_{i}(\xi_{i+1})\pi_{i+1}(\xi_i)}{\pi_{i}(\xi_{i})\pi_{i+1}(\xi_{i+1})}\Bigg)\quad \text{where}\\
\pi_i(\xi) &= e^{-0.45 (i/8) J_c(\xi)}\tau^{-i/8}.
\end{align*}
For each $i<8$, we randomly sample $10^5$ pairs of plans from the batches of $i$ and $(i+1)$ parameter values from the second half of the walks (i.e. after $500{,}000$ proposals).  We compute the acceptance probability of swapping the two states.  As a heuristic for mixing, we examine probability of accepting a proposed swap.  We find that, on average, 23.9\% of the swap proposals will be accepted.  All adjacent chains have more than a 13\% chance of accepting a proposed swap with the exception of  the $i=4$ and $5$ chains which have a 6.1\% chance of accepting a proposed swap.  Although implementing a tempering algorithm is beyond the scope of this work, these final investigations demonstrate that such a scheme should be reasonable to implement.

\section{Discussion}
Graph partition problems have had drawn long standing interest in mathematics which has only grown over recent decades.  These problems typically involve optimizing partitions according to some objective function.  More recently, however, there has been growing interest in how to sample the space of graph partitions when this space is coupled to a probability distribution. To this end, we have developed a Metropolized Forest ReCom algorithm capable of sampling graph partitions from an invariant measure. 

Our algorithm will likely have favorable mixing properties as the transition proposal will entirely redraw adjacent pairs of districts; the proposal is accepted based primarily on comparing the product of the number trees on the partition, $\tau(\xi(T))$ with $\tau(\xi(T'))$, as well as the relative effective boundaries between the altered districts.
As in \cite{moonVa, deford2019redistricting,DeFord2018}, our proposal chain relies upon (i) uniformly sampling spanning trees on a simply connected subgraph derived from adjacent partitions, and (ii) counting the total number of spanning trees on the subgraph. These two elements are the most expensive part of the proposal and yield a computational complexity that is polynomial in the number of nodes. The complexity is reduced when sampling the uniform measure on the extended tree space as we do not need to consider the number of possible trees that may be drawn on a partition. We have written a python code base that is available at \url{https://git.math.duke.edu/gitlab/gjh/mergesplitcodebase.git}

We have demonstrated the efficacy of our method on two real world problems associated with redistricting: the three state house districts of the Duplin-Onslow county cluster in North Carolina, and the 13 North Carolina congressional districts. In the former (smaller) problem, we have demonstrated good mixing with $\gamma = 1$.  In the latter we have demonstrated poor mixing for $\gamma = 1$, but good mixing for $\gamma = 0$.  We have also shown how convergence varies as $\gamma$ is adjusted from $\gamma = 0$ to $\gamma = 1$.  In principle, one could use tempering to sample at $\gamma = 1$ while using the mixing properties of $\gamma = 0$ to decorrelate the chains. This is important, because in many applications the relevant measures are based partitions rather than spanning forests. 

Although partitioning preferences may be encoded within the energy function $J(\xi)$, the Metropolized Forest ReCom procedure may not be able to efficiently explore certain constraints. However, we have provided an example where it is likely possible to use parallel or simulated tempering to sample partitions in a way that is independent of the underlying tree counts.

Our algorithm has focused on the use of uniform spanning trees on merged pairs of adjacent partitions.  However, it is not dependent on these choices and there are a number of immediate and intriguing extensions including the use of weighted edges to draw non-uniform spanning trees on the space or by merging three adjacent partitions to promote faster mixing. Furthermore, one could adapt a multi-scale framework that strictly preserves higher level structures.\footnote{A preprint of this work is available \cite{autry2020multiscale}}  Such extensions go beyond the goals of the current paper but are promising avenues toward more sophisticated sampling procedures on graph partitions.


\textbf{Acknowledgements:} We thank for Colin Rundel, Adam Graham-Squire, and Stephen Schecter for discussion around enumeration and providing the enumerations used here. We thank Dan Teague, Dana Randall, Debmalya Panigrahi, and Daryl DeFord for useful discussion. This work was initiated through the Discovering Research Mathematics Program between Duke Mathematics and the NC School of Science and Mathematics. JCM and GJH thank the NSF grant DMS-1613337 for partial support of this work.  We thank both SAMSI (DMS-1638521), TRIPODS (CFF-1934964), and the Rhodes Information Intiative for hosting workshops that were important in the development and dissemination of this work. We also thank the Rhodes Information Intiative and the Duke Provost and Deans office for financial support and creating a productive working environment.


\bibliographystyle{siamplain}
\bibliography{biblio}

\appendix

\section{Algorithmic details}

\subsection{Finding Edges to Cut on a Merged Spanning Tree}

We begin by discussing how to find the possible edges to cut from the merged spanning tree formed from the induced subgraph on the partition pair, $\xi_{ij}$.
At a high level, finding all possible edges to cut can be done in two steps. First, direct the tree so that each vertex has out-degree equal to 1, except one vertex called the \textit{root}. Second, starting at the leaves and ending at the root, compute the size of the two subgraphs that would be formed upon removing each edge and identify the edges that may be cut. These sizes can be written in terms of the sizes of the edges ``below" an edge on the rooted tree, plus the population of the connecting vertex.

The first step can be done quickly by choosing the root, then propagating outwards to the leaves. We choose the root $v$, arbitrarily and then direct all edges incident to $v$ towards $v$. We then direct the tree with a breadth first search on the neighboring vertices with edges that have not yet been oriented. For each vertex in the search, direct all of the as-of-yet undirected incident edges towards the vertex, taking note of the new neighboring vertices; repeat until all vertices are accounted for.

Determining this rooted tree step requires identifying the neighbors of each vertex once and choosing the direction for each edge once, which leads to complexity $O(|V|+|E|)$. However, $|E|=|V-1|$ on a tree, so this is $O(|V|)$.

The second step is essentially the opposite, starting from the leaves and working back to the root. For each edge incident to a leaf, define the size of the edge to be the population of the leaf.
For each edge, record the vertex closest to the root. Let $e = (v, v_{new})$ be the outgoing edge from one of these recorded vertices $v$ (i.e. the path toward the root), and $e_1, e_2, \dots$ the ingoing edges (i.e. data propagating from the leaves). The size of $e$ is
\[ \textrm{size}(e)=\text{pop}(v)+\sum\textrm{size}(e_i). \] Add the vertex $v_{new}$ to a second list of vertices; do this for all the vertices in the first list. Repeat the above step for the new list of vertices, and so on until the sizes of all edges are calculated.
If ever the root appears in one of the lists, ignore it as it does not have an outgoing edge. As the algorithm propagates, record the edges that can be cut; these are the edges that have both $\textrm{size}(e)$ and $(\textrm{pop}(G)-\textrm{size}(e))$ within the constraints specified by $J$.

Searching from the leaves to the root requires defining the size of each edge once and using the size of most of the edges in a sum once, as well as putting each vertex in a list, for a total $O(|V|+|E|)=O(|V|)$. Therefore, the whole algorithm is $O(|V|)$.

\subsection{Finding Edges to Cut When Joining Two Spanning Trees with an Edge}

For each edge that connects the two trees $T_i$ and $T_j$, we must compute the probability of choosing that edge to cut. In doing this, we must first determine every possible edge that could have been cut in the tree $T_{(T_i, T_j, e)}$, where $e\in E(T_i, T_j)$.
Na\"ively, one could simply repeat the algorithm presented above for each adjoining edge. The complexity of such an algorithm would be $O(|E(T_i, T_j)|\times |V|)$. Because we are dealing with planar graphs, it may be reasonable to assume that $|E(T_i, T_j)| \propto \sqrt{|V|}$ in which case the algorithm may scale like $O(|V|^{3/2})$. 

Recomputing the edge weights for each possible spanning tree, specified by $e\in E(T_i, T_j)$  is slow in practice and we present a faster method. 
First, we compute the edge weights once for the two split spanning trees, $T_i$ and $T_j$.
For each edge $e$ that connects the two trees, we find the edges in $T_i$ that shares a node with $e$.
For each adjacent edge in $T_i$, we may quickly find the population on each either side of the edge by using the result from the previous step on tree $T_i$.  This gives the amount of population away from the arbitrary root node of $T_i$. Thus removing some edge $e_i$ in $T_i$ will split the tree $T_{(T_i, T_j, e)}$ into a subtree with population
\begin{align}
\text{pop}_{cut}(e_i; T_{(T_i, T_j, e)}) = 
\begin{cases}
\textrm{size}(e_i) + \text{pop}(T_j) & \text{the vertex of $e$ in $T_i$}\\
& \text{is upstream from $e_i$; $e_i\in T_i$}\\
\text{pop}(T_i) - \textrm{size}(e_i) + \text{pop}(T_j) & \text{the vertex of $e$ in $T_i$ is}\\
& \text{downstream from $e_i$; $e_i\in T_i$}
\end{cases}.
\end{align}
The second subtree will have population $\text{pop}(\xi_{ij}) - \text{pop}_{cut}(e_i; T_{(T_i, T_j, e)})$. 
If splitting the tree at the edge does not satisfy population constraints, we truncate the search path. If, however, the population would be satisfied, we add the edge to the count of possible edges to cut and then consider all of its adjacent edges (independent of their directions). 
We continue to search in this way until all paths are truncated or we reach the end of the adjacent edge path along $T_i$.

We then repeat the above process examining edges $e_j\in T_j$, again starting from the edges that are connected to a node in $e$. Although this algorithm has the same complexity as the previous one, truncating the search tree makes it more efficient in practice. 

\subsection{Details on Wilson's Algorithm}

Wilson's Algorithm generates a spanning tree on a graph $G=(V, E)$ uniformly at random \cite{wilson2010dimension}. There are several possible implementations, and one is described below. All implementations use \textit{loop-erased random walks}.

\begin{enumerate}
    \item Choose two vertices $v_1$ and $u$ arbitrarily.
    \item Starting at $v_1$, walk to a neighbor $v_2$ chosen uniformly at random. Then walk from $v_2$ to one of its neighbors $v_3$ uniformly at random, and so on. If a vertex $v_i$ is reached twice, ``erase" all vertices between the two appearances of $v_i$, along with one copy of $v_i$. For example, if the sequence was $v_1, v_2, v_3, v_4, v_2$, it would be replaced with just $v_1, v_2$.
    \item Continue this loop-erased random walk until the vertex $u$ is reached. At that point, ``freeze'' the final traversed path as part of the tree.
    \item Now, choose another vertex $w$ that has not been reached arbitrarily.
    \item Perform a loop-erased random walk starting from $w$ until any vertex that is already frozen is reached.
    \item Freeze the new path.
    \item Choose another vertex $x$ and repeat the same process that was done for $w$, and so on until all vertices are frozen.
\end{enumerate}

Clearly, this gives a spanning tree of the graph. However, it is not clear that it gives one uniformly at random. The typical proof of this fact re-states Wilson's algorithm in terms of another algorithm called ``cycle popping,'' and it is proved that that algorithm generates a uniform spanning tree, so Wilson's Algorithm does, too.

How fast is Wilson's Algorithm? Na\"ively, there is about a $n/(V-1)$ chance each step of stepping to one of $n$ vertices out of the $V$ vertices of the graph (this is exactly true on a complete graph). Then the first part of the algorithm, loop-erased random walk from $v$ to $u$, thus takes about $V$ steps. It is harder to estimate the number of vertices that are reached from a loop-erased random walk, but a reasonable assumption is that it is approximately some proportion $k$ of all the vertices in the graph, with the proportion depending on the graph structure. Then, the walk from $w$ to one of these vertices takes about $kV$ steps and incorporates $k^2$ of the remaining vertices in the graph. Then the next vertex takes $k^2V$ time, then $k^3V$, and so on, finishing in about $\log_{k}(1/V)$ steps, so our total time is about

\begin{align*}
    V+kV+k^2V+\dots+k^{\log_{k}(1/V)}V&=V\frac{1-k^{1-\log_{k}V}}{1-k} \\
    &=\frac{V-k}{1-k} \\
    &\approx V\cdot \text{number of iterations needed} \\
    &= O(V\log V) \text{ on some graphs} \\
    &\le O(V^2)
\end{align*}

In fact, Wilson's algorithm expected runtime is the graph's mean hitting time \cite{wilsonGeneratingRandomSpanning1996}, which is upper bounded by the graph's cover time, which is in turn bounded by $O(V^2)$ in planar graphs \cite{jonassonCoverTimePlanar2000}. This is not (asymptotically) the slowest step as seen in Appendix~\ref{appendix:kirchhoff}.

\subsection{Calculating \texorpdfstring{$\tau$}{tau} Using Kirchhoff's Theorem}
\label{appendix:kirchhoff}

A remarkable theorem is that the number of spanning trees on a graph $G$ may be computed as

\[ \tau(G) = \det{Q^*} \]
where $Q^*$ is any minor of $Q$, the Laplacian matrix on $G$. The Laplacian matrix is equal to $A-D$, where $A$ is the adjacency matrix of $G$ and $D$ is the degree matrix, a diagonal matrix whose entries are vertex degrees.

Constructing $Q^*$ is $O(V+E)$ ($=O(V)$ for planar graphs like ours), and computing the determinant of an $n\times n$ matrix is $O(n^{2.373})$ (though most reasonable implementations are slower) \cite{williamsMultiplyingMatricesN2} \cite{gallPowersTensorsFast2014}. This is, asymptotically, the slowest step.

In practice, we compute the logarithms of these determinants, add and subtract them, then exponentiate, since the actual values of $\tau$ can get extremely large, and we can run into overflow errors otherwise. This does not affect the time complexity.

\end{document}